\documentclass[reprint,aps,pra,floatfix,superscriptaddress]{revtex4-2}
\usepackage{amsmath,amssymb,bm}
\usepackage{graphicx}
\usepackage{subcaption}
\usepackage{booktabs}
\usepackage{multirow}
\usepackage{siunitx}
\usepackage{hyperref}
\usepackage{microtype}
\usepackage{amsthm}
\usepackage[section]{placeins}
\usepackage{bibunits}
\usepackage{xcolor}
\usepackage[version=4]{mhchem}
\defaultbibliographystyle{unsrt}
\newcommand{\R}{\mathbb{R}}

\newtheorem{assumption}{Assumption}
\newtheorem{lemma}{Lemma}
\newtheorem{proposition}{Proposition}
\newtheorem{theorem}{Theorem}
\newtheorem{corollary}{Corollary}
\newtheorem{remark}{Remark}
\graphicspath{{./}}

\begin{document}
\begin{bibunit}

\title{Photonic Exponential Approximation via Cascaded TFLN Microring Resonators toward Softmax}
\author{Hyoseok Park}
\affiliation{Department of Physics, Chungnam National University, Daejeon 34134, Republic of Korea}
\author{Yeonsang Park}
\email{yeonsang.park@cnu.ac.kr}
\thanks{Corresponding author}
\affiliation{Department of Physics, Chungnam National University, Daejeon 34134, Republic of Korea}
\date{\today}
\begin{abstract}
The rapid growth of large-scale AI models has intensified energy consumption and data-movement challenges in modern datacenters.
Photonic accelerators offer a promising path by executing the linear matrix multiplications of transformer inference at high throughput and low energy.
However, the softmax attention layer---which requires element-wise exponentiation followed by normalization---still relies on electronic post-processing, creating an electro-optic conversion bottleneck that negates much of the potential photonic advantage.

We present a cascaded micro-ring resonator (MRR) architecture that synthesizes the per-channel exponential function required by softmax, $e^{x_n-\max(x)}$, over a finite interval with tunable worst-case relative error.
A control signal detunes each ring via an electro-optic mechanism;
a weak probe at fixed frequency experiences Lorentzian transmission, and cascading $N$ identical stages yields a multiplicative transfer function whose logarithm is approximately linear.

We derive mapping rules, depth-scaling estimates, and a minimax fitting formulation, and validate the framework with three-dimensional FDTD simulations of X-cut thin-film lithium niobate (TFLN) add-drop micro-ring resonators.
Direct multi-ring FDTD validation extends to a five-ring cascade and confirms agreement with theory primarily over the upper operating range; deeper cascades and higher quality factors are assessed analytically.
The cascade implements the per-channel exponential block---the key missing nonlinearity for photonic softmax.
We further present a WDM-parallel chip architecture with closed-loop PI feedback that completes the full softmax---exponentiation, summation, and normalization---on a single photonic chip without per-channel normalization circuitry.
\end{abstract}

\maketitle

\section{Introduction}
Transformer inference is often limited by power and memory traffic, motivating optical accelerators that exploit parallel propagation and multiplexing~\cite{vaswani2017attention,dao2022flashattention,shen2017deep,feldmann2021parallel,miscuglio2020photonic,li2023photonictensorcore}.
Recent perspective articles also discuss data-center power consumption as one motivation for optical computing~\cite{savage2025light,wang2024siliconroadmap}.
While linear operators are comparatively amenable to photonic implementation~\cite{shen2017deep,feldmann2021parallel,harris2018linear}, the softmax function used in attention layers requires an exponential mapping together with global normalization---both difficult to realize in passive photonic circuits, where transmission is fundamentally bounded by unity.
Parallel digital-hardware studies treat the exponential/softmax stage as a bottleneck and propose dedicated approximations~\cite{softermax2021,du2023softmaxhw,sole2023,base2softmax2022,tcasii_teas2023,tvlsi_softmax2022,dash2025softonic,zhan2024oe_photonicsoftmax,photonix2025_transformerchip}.
Many integrated-photonic classifier demonstrations still rely on electronic post-processing for the final nonlinear readout~\cite{hu2025lnsoftmax_electronic}; the resulting electro-optic conversion overhead can negate the throughput and energy benefits of the photonic front-end.
Notably, the SOFTONIC architecture~\cite{dash2025softonic} explicitly argues that ``the inability of MRRs and MZMs to handle SMA's exponential and division functions'' necessitates alternative approaches based on microdisk modulators and polynomial approximation, achieving 89.7\% accuracy with a third-degree Chebyshev polynomial.
Here we challenge this premise: we show that a passive Lorentzian cascade of microring resonators can be tuned so that its logarithm is approximately linear over a finite interval, enabling exponential-function synthesis with sub-2\% worst-case error---an order of magnitude more accurate than SOFTONIC's polynomial approach---while remaining compatible with integrated microring platforms~\cite{bogaerts2012silicon,heebner2008optical,zhang2023microringonn,li2022starlight,jang2022mrractivation}.
We term this cascade block an \emph{approximate exponential function} (AEF) unit.
We further propose a WDM-parallel architecture with a single PI feedback loop that realizes the complete softmax function---including summation and normalization---without per-channel electronic processing.

We extend the theoretical framework with three-dimensional FDTD simulations of a single X-cut TFLN add-drop micro-ring resonator. The simulated device parameters---quality factor, free spectral range, and electro-optic sensitivity---calibrate the cascade design parameters, bridging analytical fitting and physically realizable hardware.
Two operating regimes emerge from this calibration: an FDTD-characterized regime with moderate drop-port depth ($D_{\max}\approx 0.36$), where the analytic error stays below ${\sim}5\%$ for $N\le 7$ but the power budget limits practical cascades to $N\le 5$; and a projected high-$Q$ regime ($D_{\max}\ge 0.95$), enabling deeper cascades ($N\le 30$) with sub-percent error.
Cascade performance is predicted analytically and validated by a five-ring cascade 3D~FDTD simulation (Sec.~\ref{sec:device}).

The paper is organized as follows: Section~II presents the mapping, transfer model, and depth-design rules; Section~III provides numerical fits and validation; Section~IV describes the single-ring TFLN device design and FDTD validation; Section~V assesses physical feasibility including voltage requirements, insertion loss, and energy efficiency; Section~VI discusses implementation scope, platform comparisons, and limits; and Section~VII concludes.

\section{Model and Design Framework}

\begin{figure*}[t]
\centering
\includegraphics[width=\textwidth]{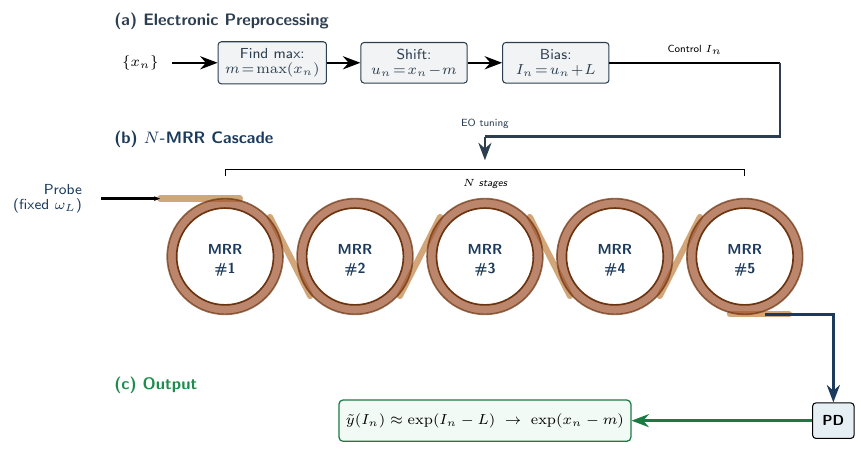}
\caption{Overview of the control--probe add-drop cascade $N$-MRR exponential block. (a) Electronic preprocessing maps an arbitrary input sequence $\{x_n\}$ to a nonnegative control signal via $m=\max_n x_n$, $u_n=x_n-m$, and $I_n=u_n+L$ with $L=m-\min_n x_n$. (b) The control signal $I_n$ induces resonance shifts in a cascade of $N$ rings, while a weak fixed-frequency probe propagates through the serial add-drop cascade (the drop output of each ring feeds the next stage), experiencing multiplicative transmission. (c) After photodetection, the block implements $y(I_n)\approx \exp(I_n-L)\approx \exp(x_n-m)$, i.e., the normalized exponential used in softmax.}
\label{fig:system}
\end{figure*}
\noindent\textbf{Target mapping.}
Let $x=(x_1,\dots,x_K)\in\R^K$ be an arbitrary real-valued sequence (or vector).
Directly generating $\exp(x_n)$ as a \emph{passive} optical transmission is impossible in general because $\exp(x)$ grows beyond unity while a passive transmission satisfies $0<T\le 1$~\cite{saleh2007fundamentals}.
However, for softmax,

\begin{equation}
\mathrm{softmax}(x)_n=\frac{e^{x_n}}{\sum_j e^{x_j}},
\end{equation}
a common shift cancels:

\begin{equation}
\frac{e^{x_n+c}}{\sum_j e^{x_j+c}}=\frac{e^{x_n}}{\sum_j e^{x_j}}
\qquad(\forall c\in\R).
\label{eq:shiftinv}
\end{equation}
Thus it suffices to generate

\begin{equation}
e^{x_n-m},\qquad m\equiv \max_j x_j,
\label{eq:need}
\end{equation}
since the global factor $e^m$ cancels.

To ensure a nonnegative control-signal amplitude, define

\begin{equation}
u_n\equiv x_n-m\le 0,\qquad
L\equiv -\min_n u_n = m-\min_n x_n\ge 0,
\label{eq:defL}
\end{equation}
and map each scalar to a nonnegative control-signal amplitude

\begin{equation}
I_n\equiv u_n+L \in [0,L].
\label{eq:Ik}
\end{equation}
Then

\begin{equation}
e^{x_n-m}=e^{u_n}=e^{I_n-L}.
\label{eq:targetNormalized}
\end{equation}
Hence the optical design task is to realize, for $I\in[0,L]$,

\begin{equation}
f(I)=e^{I-L}\in[e^{-L},1].
\label{eq:target}
\end{equation}

\noindent\textbf{Control--probe transfer.}
Consider a weak probe at fixed angular frequency $\omega_L$.
For the $k$th ring, let $\omega_{0,k}$ denote its resonance frequency and $\Gamma>0$ its loaded half-width at half maximum (HWHM).
Define the detuning

\begin{equation}
\Delta\omega_k \equiv \omega_L-\omega_{0,k}.
\end{equation}
Near resonance, the normalized Lorentzian transmission is modeled as~\cite{bogaerts2012silicon,heebner2008optical}

\begin{equation}
T_k(\Delta\omega_k)=\frac{1}{1+\left(\frac{\Delta\omega_k}{\Gamma}\right)^2}.
\label{eq:Lorentz}
\end{equation}
In a control--probe architecture, a nonnegative control-signal amplitude $I\ge 0$ shifts the ring resonance. Here $I$ denotes a generic control amplitude: for optical-pump operation it maps to optical intensity, while for EO operation it maps to electrical control level (e.g., voltage).
Across many physical mechanisms (optical pump via Kerr/XPM, EO drive via Pockels effect, thermal, carrier tuning), the shift can be linearized on a working range~\cite{almeida2004alloptical,xu2005micrometre,bogaerts2012silicon,padmaraju2014thermal,hu2023carrierlinearity,resonancetrim2021}:

\begin{equation}
\omega_{0,k}(I)=\omega_{0,k}^{(0)}+\eta I,
\label{eq:shift}
\end{equation}
where $\omega_{0,k}^{(0)}$ is the cold-cavity resonance and $\eta$ is the control-to-resonance sensitivity.
In practice, the control channel can be optical or electrical (optical pump, EO/Pockels drive, thermal, or carrier tuning); a quantitative EO feasibility example is given in the Discussion.
With $\Delta\omega_{0,k}\equiv \omega_L-\omega_{0,k}^{(0)}$, the control-dependent detuning becomes

\begin{equation}
\Delta\omega_k(I)=\Delta\omega_{0,k}-\eta I.
\end{equation}
Define dimensionless parameters

\begin{equation}
a_k\equiv \frac{\Delta\omega_{0,k}}{\Gamma},
\qquad
b\equiv -\frac{\eta}{\Gamma}.
\label{eq:abdef}
\end{equation}
Then Eq.~(\ref{eq:Lorentz}) yields the control-to-probe transfer of a single ring,

\begin{equation}
T_k(I)=\frac{1}{1+(a_k+bI)^2}.
\label{eq:Tk}
\end{equation}
Physical meaning: $a_k$ is a \emph{static detuning} in linewidth units (set by heater/carrier tuning/fabrication),
and $|b|$ is the \emph{normalized sensitivity magnitude} (linewidths of resonance shift per unit control-signal amplitude); the sign convention is absorbed into the detuning expression.
For ``same-material/same-geometry'' rings, $b$ is often common, while $a_k$ can be tuned per ring.

\noindent\emph{Sign convention.} Simultaneously flipping $(a_k,b)\mapsto(-a_k,-b)$ leaves $T_k(I)$ unchanged, so we may take $b>0$ without loss of generality.


Let $N$ rings be cascaded in a serial add-drop topology: $T_k(I)$ denotes
the add-to-drop transmission of ring~$k$, and the drop output of ring~$k$
feeds the add (input bus) port of ring~$k\!+\!1$.
Assuming the probe is sufficiently weak so the control channel dominates the resonance shift, the normalized probe output is the product

\begin{equation}
y(I)\equiv \frac{P_\mathrm{out}^{(\mathrm{probe})}(I)}{P_\mathrm{in}^{(\mathrm{probe})}}
= \prod_{k=1}^{N} T_k(I)
=
\prod_{k=1}^{N}\frac{1}{1+(a_k+bI)^2}.
\label{eq:y}
\end{equation}
To focus on the \emph{shape} of the approximation, we allow a global scale factor $C>0$:

\begin{equation}
\tilde y(I)\equiv C\,y(I).
\label{eq:scale}
\end{equation}
In softmax, $p_n = Ce^{I_n-L}/\sum_j Ce^{I_j-L}$, so $C$ cancels between numerator and denominator and is physically inessential; nevertheless it is convenient for error analysis.
For a fixed $(N,b,\{a_k\})$, the optimal $C$ for the minimax log-error in Eq.~(\ref{eq:Einf}) can be written in closed form.
Let $g(I)\equiv \ln y(I)-(I-L)$ on $[0,L]$.
Then the minimax-optimal shift is $\ln C^\star = -( \max_I g(I) + \min_I g(I))/2$, yielding
$E_\infty = (\max_I g(I) - \min_I g(I))/2$.

Taking logarithms,

\begin{equation}
\ln \tilde y(I)=\ln C -\sum_{k=1}^{N}\ln\!\big(1+(a_k+bI)^2\big).
\label{eq:logy}
\end{equation}
The target $\ln f(I)=I-L$ is linear; hence exponential approximation is equivalent to the \emph{log-linearization} goal

\begin{equation}
\ln \tilde y(I)\approx I-L
\qquad \text{uniformly on } I\in[0,L].
\label{eq:logfit}
\end{equation}

\noindent\textbf{Error metric.}
Define the worst-case \emph{log-error} on $[0,L]$:

\begin{equation}
E_\infty \equiv \sup_{I\in[0,L]}\Big|\ln \tilde y(I)-(I-L)\Big|.
\label{eq:Einf}
\end{equation}
If $E_\infty\le \varepsilon_{\log}$, then for all $I\in[0,L]$,

\begin{equation}
e^{-\varepsilon_{\log}}\le \frac{\tilde y(I)}{f(I)} \le e^{\varepsilon_{\log}}
\Rightarrow
\left|\frac{\tilde y(I)}{f(I)}-1\right|\le e^{\varepsilon_{\log}}-1.
\label{eq:rel_from_log}
\end{equation}
Thus achieving a prescribed worst-case relative error $\varepsilon$ is guaranteed by

\begin{equation}
E_\infty \le \varepsilon_{\log}\equiv \ln(1+\varepsilon)\approx \varepsilon.
\label{eq:epslog}
\end{equation}
\noindent\textbf{Depth scaling.}
We derive depth-related constraints and design rules for a prescribed approximation tolerance.

\noindent\textbf{Necessary slope condition.}
Differentiate Eq.~(\ref{eq:logy}):

\begin{equation}
\frac{d}{dI}\ln y(I)=
-\sum_{k=1}^N
\frac{2b(a_k+bI)}{1+(a_k+bI)^2}.
\end{equation}
Since $|2u/(1+u^2)|\le 1$ for all real $u$,

\begin{equation}
\left|\frac{d}{dI}\ln y(I)\right|\le N|b|.
\label{eq:slopebound}
\end{equation}
The target $\ln f(I)=I-L$ has constant slope $+1$, so a necessary condition to track it is

\begin{equation}
N|b| \gtrsim 1.
\label{eq:Nb}
\end{equation}
\textbf{Near-optimal parameterization.}
The full design problem can be written as a minimax fit in the log domain~\cite{cheney1966approximation}:

\begin{equation}
\begin{gathered}
\min_{a_1,\dots,a_N,\ \ln C}\;\sup_{I\in[0,L]}\,|r(I)|,\\
r(I)\equiv \ln C-\sum_{k=1}^{N}\ln\!\big(1+(a_k+bI)^2\big)-(I-L).
\end{gathered}
\label{eq:minimax}
\end{equation}
This objective is permutation-invariant in the $a_k$'s (ring index $k$).
In practice (and in numerical experiments reported below), the optimizer frequently collapses to a permutation-symmetric solution

\begin{equation}
a_1=\cdots=a_N\equiv a,
\label{eq:identical}
\end{equation}
reducing the design to two parameters $(a,b)$ (plus $C$).
With Eq.~(\ref{eq:identical}),

\begin{equation}
\tilde y(I)=C\,y(I)=C\left[\frac{1}{1+(a+bI)^2}\right]^N.
\end{equation}
A robust initialization is obtained by placing the midpoint of the interval on the Lorentzian half-maximum flank and matching the slope:

\begin{equation}
a+b\frac{L}{2}\approx -1,
\qquad
N b \approx 1.
\label{eq:flankrules}
\end{equation}
These two equations already yield a good design; a small (two-parameter) refinement can then enforce the desired worst-case tolerance.

\textbf{Local expansion and depth scaling.}
A Taylor expansion of the log-domain residual around the flank-centered point $I_0=L/2$ (with $a+bI_0=-1$ and $Nb=1$) shows that the quadratic term vanishes identically, leaving a leading cubic residual $r(\delta)\sim \delta^3/(6N^2)$.
Over $I\in[0,L]$, this implies $E_\infty\sim L^3/N^2$, so that achieving a prescribed tolerance $\varepsilon_{\log}$ requires $N\propto L^{3/2}/\sqrt{\varepsilon_{\log}}$, which explains the scaling used in Eq.~(\ref{eq:Nestimate}).
The full derivation is provided in Supplementary Sec.~S0; an intuitive local-expansion summary appears in Sec.~S1.

\textbf{Practical engineering estimate.}
Given $L$ and a target worst-case relative error $\varepsilon$, define $\varepsilon_{\log}=\ln(1+\varepsilon)$.
A heuristic engineering estimate (not a rigorous bound) that matched our percent-level numerical designs is
\begin{equation}
\boxed{
N \ \approx\
\left\lceil
\max\left(\frac{1}{b_{\max}},\ \kappa \frac{L^{3/2}}{\sqrt{\varepsilon_{\log}}}\right)
\right\rceil,
}
\label{eq:Nestimate}
\end{equation}
where $b_{\max}$ is the physically achievable sensitivity bound and $\kappa\simeq 0.07$ for the identical-detuning flank design with a minimax refinement.
After choosing $N$, set $b=\min(b_{\max},1/N)$ and $a=-1-bL/2$ as initialization, then refine $(a,b)$ by a two-parameter minimax fit on $[0,L]$.

A heuristic conservative screening bound $N \ge \lceil(L^2/4 + 1/(2b^2))/\ln(1+\varepsilon)\rceil$ (derived via the same local-expansion argument; see Supplementary Sec.~S1) provides a quick upper estimate but is not a rigorous guarantee.


\section{Numerical Fits and Validation}
We validate the analytical framework with minimax numerical fits and sampled robustness checks.
Figure~\ref{fig:cascade_fit} shows the fitted approximation quality at $L=8$:
the top (linear) panel plots $N=1,3,5,7$ over $I\in[0,20]$, the middle (log) panel compares $N=5,10,20,30$ on $I\in[0,8]$, and the bottom panel shows the pointwise relative error with the characteristic Chebyshev equioscillation pattern.

\begin{figure}[!t]
\centering
\includegraphics[width=\columnwidth,height=0.75\textheight,keepaspectratio]{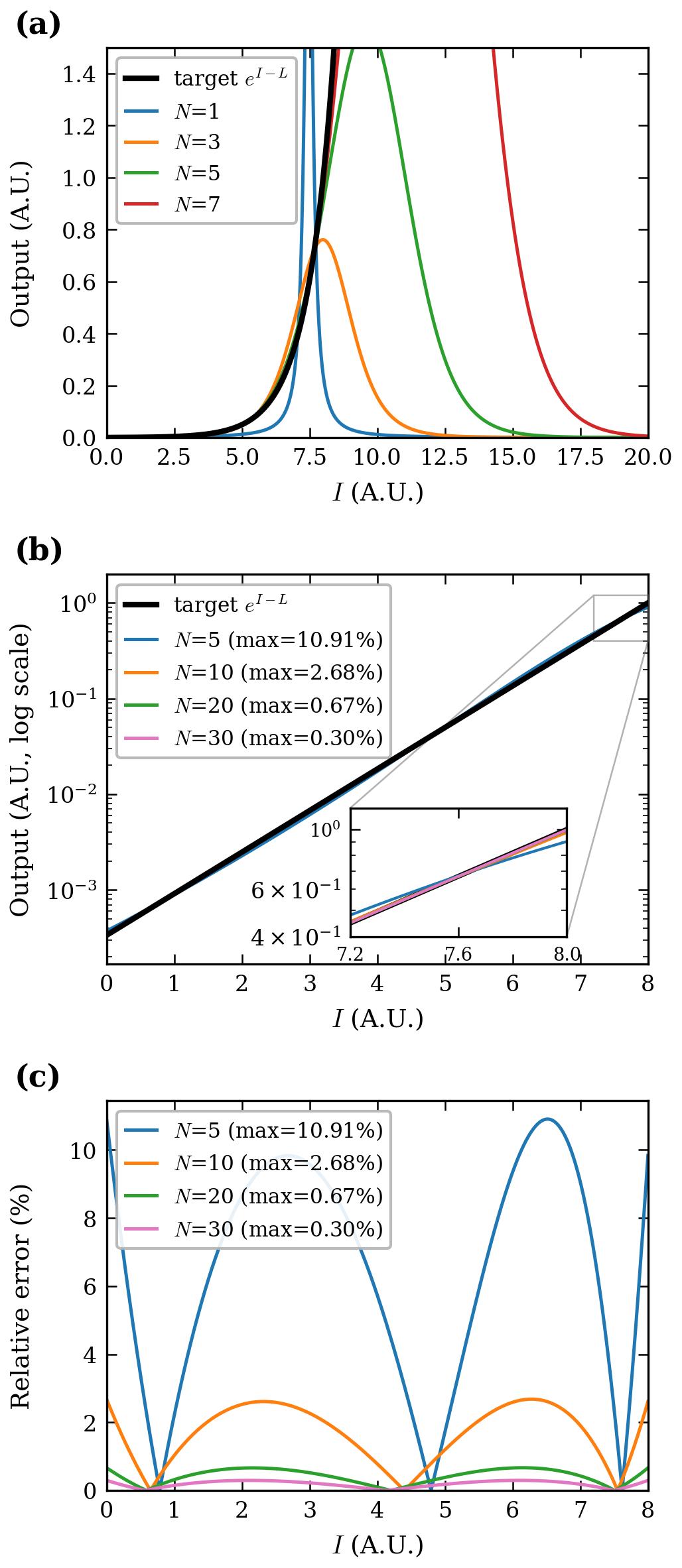}
\caption{Minimax cascade fits at $L = 8$.
(a)~Linear scale: shallow cascades ($N = 1, 3, 5, 7$) over
$I \in [0, 20]$.  The target $e^{I-L}$ (black) is progressively
better matched as $N$ increases.
(b)~Log scale: depth comparison ($N = 5, 10, 20, 30$) on
$I \in [0, 8]$.  Inset zooms into $I \in [6, 8]$ showing convergence.
(c)~Pointwise relative error showing the Chebyshev
equioscillation pattern characteristic of minimax optimality.}
\label{fig:cascade_fit}
\end{figure}

We fit identical-detuning cascades (Eq.~\ref{eq:identical}) on $I\in[0,L]$ and compare several depths using a minimax criterion.

Table~\ref{tab:Ncompare} makes the accuracy--depth trade-off explicit at $L=8$.
A worked input-to-output example demonstrating the mapping from an arbitrary input sequence $x=[-3.2, 1.2, 4.8, -0.9]$ through the cascade is provided in Supplementary Sec.~S2.
The example shows that the $N=10$ cascade keeps the worst-case relative error below $2.7\%$ across all channels.

\begin{table}[tbp]
\caption{Depth comparison for $L=8$ using fitted $\tilde y(I)=C[1+(a+bI)^2]^{-N}$ (same fitting pipeline for all $N$).}
\label{tab:Ncompare}
\begin{ruledtabular}
\begin{tabular}{rcccc}
$N$ & $a$ & $b$ & max rel. err. & mean rel. err. \\
\midrule
5  & $-2.0789$ & $0.21658$ & $10.9\%$ & $6.43\%$ \\
10 & $-1.4588$ & $0.10202$ & $2.68\%$ & $1.65\%$ \\
20 & $-1.2135$ & $0.05025$ & $0.67\%$ & $0.42\%$ \\
30 & $-1.1392$ & $0.03341$ & $0.30\%$ & $0.19\%$ \\
\end{tabular}
\end{ruledtabular}
\end{table}

\noindent\textbf{Empirical calibration.}
We calibrate the effective logit range $L_\mathrm{eff}$ from autoregressive Transformers (distilgpt2/gpt2)~\cite{vaswani2017attention,radford2019gpt2,hf_distilgpt2_card,karpathy_tinyshakespeare,gutenberg_pride_prejudice} at context length 128, finding $L_{\mathrm{eff},0.999} \approx 7$--$9$ at the 50th--90th percentiles (Supplementary Sec.~S2).
A clipping threshold $t^* = -12$ preserves p99 softmax accuracy below $0.1\%$.
Full protocol details, clipping-sweep tables/plots, and per-run statistics are provided in Supplementary Sec.~S3.

A synthetic design-space map (Supplementary Table~\ref{tab:S-LLMrange}) shows that near $L \approx 8$, moderate depth ($N \ge 10$) reaches few-percent error, whereas $L \gtrsim 12$ requires deeper cascades.
All fits follow the same pipeline: minimize the worst-case log-error on a uniform grid, initialize from the flank rules in Eq.~(\ref{eq:flankrules}), perform multi-start global search, and apply bounded local refinement; implementation details and scripts are provided in a public repository~\cite{eer_aef_repo} (commit: \texttt{585e695}).

\FloatBarrier

\section{TFLN Single-Ring Device Design and FDTD Validation}
\label{sec:device}

\subsection{Waveguide and ring geometry}
\label{sec:waveguide}

The device is based on an X-cut thin-film lithium niobate (\ce{LiNbO3})
on insulator wafer with a \SI{600}{\nano\meter}-thick \ce{LiNbO3} film
on \ce{SiO2}.  A \SI{500}{\nano\meter}-deep rib etch defines a
\SI{1.4}{\micro\meter}-wide single-mode waveguide with a
\SI{100}{\nano\meter} unetched slab (Fig.~\ref{fig:cross_section}).
Lumerical MODE simulations yield $n_\mathrm{eff} = 1.903$ and
$n_g = 2.24$ at $\lambda = \SI{1550}{\nano\meter}$ for the
fundamental TE$_0$ mode.

The ring resonator ($R = \SI{20}{\micro\meter}$,
$L_\mathrm{ring} = \SI{125.7}{\micro\meter}$) is configured as an
add-drop resonator with \SI{100}{\nano\meter} coupling gaps
(Fig.~\ref{fig:topview}).  The FDTD-measured free spectral range is
$\mathrm{FSR} = \SI{8.29}{\nano\meter}$ ($n_g \approx 2.30$), slightly
above the MODE value due to bend-induced dispersion.

\begin{figure}[htbp]
\centering
\includegraphics[width=\columnwidth]{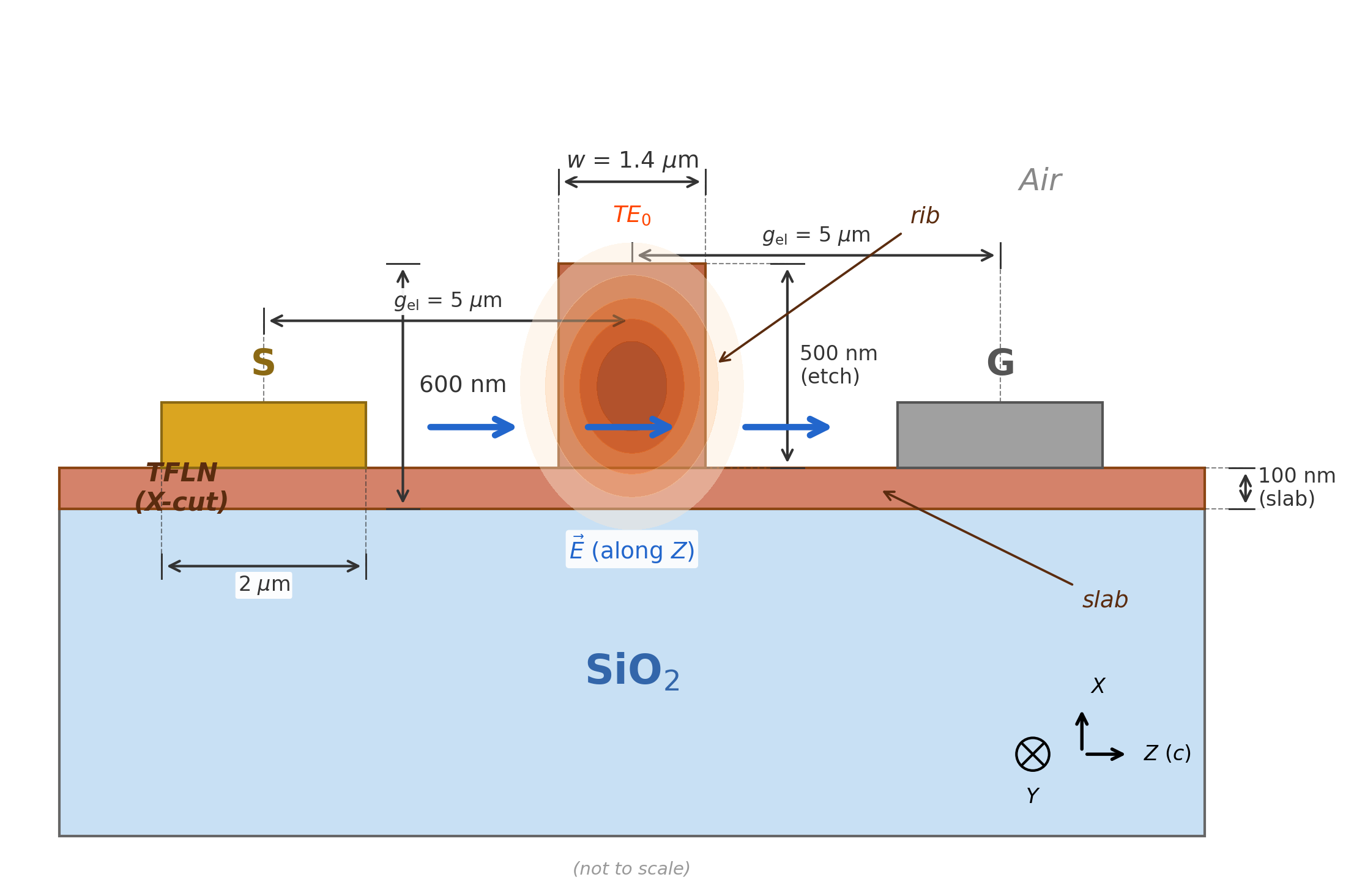}
\caption{Cross-section of the X-cut TFLN rib waveguide on a \ce{SiO2}
substrate.  The \SI{600}{\nano\meter} \ce{LiNbO3} film is etched
\SI{500}{\nano\meter} to form a \SI{1.4}{\micro\meter}-wide single-mode
rib waveguide.  Lateral signal~(S) and ground~(G) electrode positions are
indicated; electrode design details are discussed in
Sec.~\ref{sec:electrode}.}
\label{fig:cross_section}
\end{figure}

Table~\ref{tab:device_params} summarizes the waveguide and ring
parameters.

\begin{table}[htbp]
\centering
\caption{Waveguide and ring parameters of the X-cut TFLN micro-ring
resonator.  Electro-optic electrode parameters are listed separately in
Table~\ref{tab:eo_params}.}
\label{tab:device_params}
\begin{ruledtabular}
\begin{tabular}{llcl}
Parameter & Symbol & Value & Unit \\
\hline
Total TFLN thickness   & $t_\mathrm{TFLN}$       & 600   & \si{\nano\meter} \\
Etch depth             & $t_\mathrm{etch}$       & 500   & \si{\nano\meter} \\
Slab thickness         & $t_\mathrm{slab}$       & 100   & \si{\nano\meter} \\
Waveguide width        & $w$                      & 1.4   & \si{\micro\meter} \\
Bend radius            & $R$                      & 20    & \si{\micro\meter} \\
Coupling gap           & $g$                      & 100   & \si{\nano\meter} \\
Circumference          & $L_\mathrm{ring}$        & 125.7 & \si{\micro\meter} \\
Free spectral range    & FSR                      & 8.29  & \si{\nano\meter} \\
Effective index (TE$_0$) & $n_\mathrm{eff}$       & 1.903 & --- \\
Group index (TE$_0$)   & $n_g$                    & 2.24  & --- \\
Extraordinary index    & $n_e$                    & 2.138 & --- \\
\end{tabular}
\end{ruledtabular}
\end{table}

\subsection{3D FDTD Methodology}
\label{sec:fdtd_method}

The ring resonator response is simulated using Lumerical 3D FDTD with
conformal variant~1 meshing.  A broadband TE$_0$ mode source
(\SIrange{1530}{1570}{\nano\meter}) is injected into the input bus
waveguide, and through- and drop-port spectra are recorded.
A ``z-refined 3-fix'' meshing strategy ensures convergence in the
thin-film geometry~\cite{zhu2021aop}; detailed simulation setup is
provided in Supplementary Sec.~S4 (Table~\ref{tab:fdtd_params}).

\begin{figure}[htbp]
\centering
\includegraphics[width=\columnwidth]{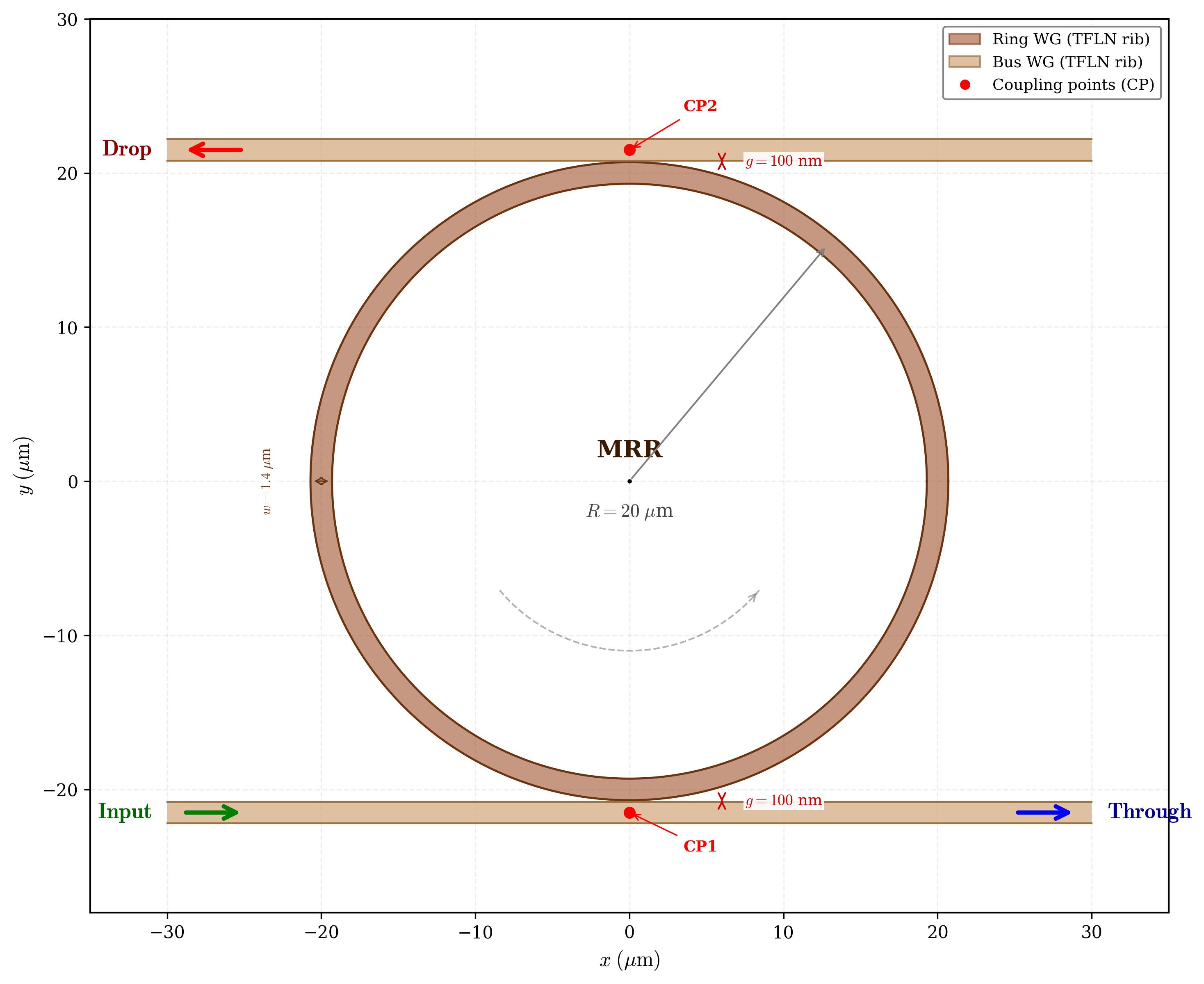}
\caption{Top view of the single add-drop micro-ring resonator used in
the 3D FDTD simulation.  The ring waveguide ($R = \SI{20}{\micro\meter}$,
$w = \SI{1.4}{\micro\meter}$) is evanescently coupled to input and drop
bus waveguides through \SI{100}{\nano\meter} gaps at coupling points
CP1 and CP2.}
\label{fig:topview}
\end{figure}

\subsection{Single-Ring Add-Drop Results}
\label{sec:single_ring}

Figure~\ref{fig:single_ring_spectrum} shows the through- and drop-port
spectra from 3D FDTD.  Five resonances are resolved across
\SIrange{1530}{1570}{\nano\meter} with
$\mathrm{FSR} = \SI{8.29}{\nano\meter}$ ($n_g \approx 2.30$).

\begin{figure}[htbp]
\centering
\includegraphics[width=\columnwidth]{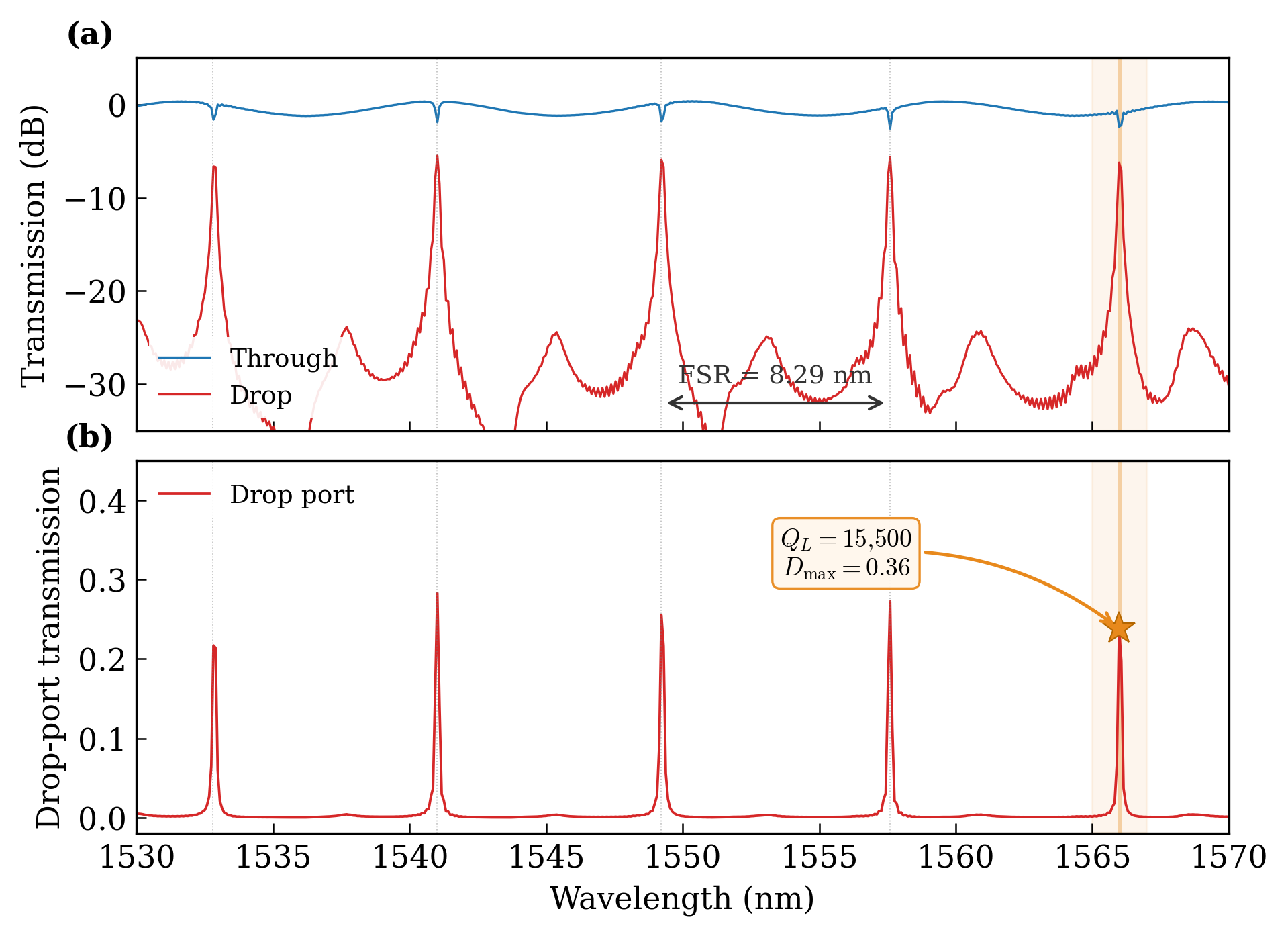}
\caption{Simulated through-port (blue) and drop-port (red) transmission
spectra of the single add-drop micro-ring resonator from 3D FDTD.
Top: logarithmic scale; bottom: linear scale.  Five resonances are
visible with $\mathrm{FSR} \approx \SI{8.29}{\nano\meter}$.}
\label{fig:single_ring_spectrum}
\end{figure}

Lorentzian fitting of the drop-port peaks yields
$Q_L = 10{,}300$--$15{,}500$, with the best resonance at
$\lambda = \SI{1566}{\nano\meter}$ reaching $Q_L = 15{,}500$
($\mathrm{FWHM} = \SI{101}{\pico\meter}$,
$D_\mathrm{max} = 0.360$, $\SI{-4.4}{\decibel}$).
The through-port extinction ratio is
\SIrange{1.6}{2.6}{\decibel}, and the five-resonance mean is
$Q_L = 12{,}500 \pm 1{,}800$ ($D_\mathrm{max} = 0.29$--$0.36$).
CMT analysis of the best resonance gives
$Q_i = Q_L/(1 - \sqrt{D_\mathrm{max}}) = 15{,}500/0.400 \approx 38{,}800$,
confirming that the \SI{500}{\nano\meter} etch provides sufficient
confinement and that the \SI{100}{\nano\meter} gap places the ring
in the coupling-limited regime.
The cascade analysis below adopts the best-case FDTD calibration
($Q_L = 15{,}500$, $D_\mathrm{max} = 0.360$); using the five-resonance
mean would increase required voltages by ${\sim}24\%$
(see Table~\ref{tab:cascade} caption).

The simulation time of \SI{50}{\pico\second} exceeds the loaded
photon lifetime $\tau_L = Q_L\lambda_0/(2\pi c) \approx \SI{12.7}{\pico\second}$
by ${\sim}4\times$, but the intrinsic lifetime
$\tau_i \approx \SI{32}{\pico\second}$ is comparable, so the extracted
$Q_i$ may be slightly conservative.
An independent eigenmode (FDE) analysis of the same cross-section
at $R = \SI{20}{\micro\meter}$---using a $300 \times 300$ mesh
($\Delta y \approx \SI{10}{\nano\meter}$, $5\times$ finer than the
FDTD vertical grid)---yields
$Q_\mathrm{rad+leak} = 2.4 \times 10^7$; including bulk \ce{LiNbO3}
absorption ($\Gamma = 0.89$) gives a theoretical
$Q_i > 10^7$~\cite{zhu2021aop, hu2024integratedeo, zhang2017highQ_optica,
zhuang2023wetetch, zhu2024highQ_racetrack, gao2022ultrahighQ},
confirming that the gap between the numerical $Q_i$ and published
values ($>10^6$) originates from mesh discretization
(Supplementary~S4.5, Table~\ref{tab:theoretical_Qi}).
In the CMT framework, $D_\mathrm{max} = [2\kappa/(2\kappa+\gamma)]^2$
increases as $Q_i$ rises; at the present coupling gap,
increasing $Q_i$ to $10^6$ would raise $D_\mathrm{max}$ from $0.36$ to
${\sim}0.95$ and $Q_L$ from $15{,}500$ to ${\sim}25{,}200$.

Figure~\ref{fig:fdtd_aef}(a) shows a Lorentzian fit to the best
drop-port resonance at $\lambda = \SI{1566}{\nano\meter}$, validating
the cascade model (Eq.~\ref{eq:Lorentz}).
Figure~\ref{fig:fdtd_aef}(b) demonstrates that cascading $N$ copies
of this FDTD-extracted Lorentzian reproduces the target exponential
$e^{I-L}$ with increasing fidelity as $N$ grows.

To validate the cascade prediction directly, a five-ring cascade
3D~FDTD simulation was performed using Tidy3D~\cite{tidy3d2024};
the full simulation notebook is publicly available~\cite{tidy3d2024}.
The $|E|^2$ field at $\lambda = \SI{1549}{\nano\meter}$
[Fig.~\ref{fig:fdtd_aef}(d)] confirms resonant excitation across
all five rings.  Mapping the drop-port spectrum onto the control
variable $I$ yields 11~data points within the AEF operating range
[Fig.~\ref{fig:fdtd_aef}(e,\,f)], with the FDTD transmission
closely tracking the $N = 5$ theoretical curve near $I \approx L = 8$.

\begin{figure*}[t]
\centering
\includegraphics[width=\textwidth]{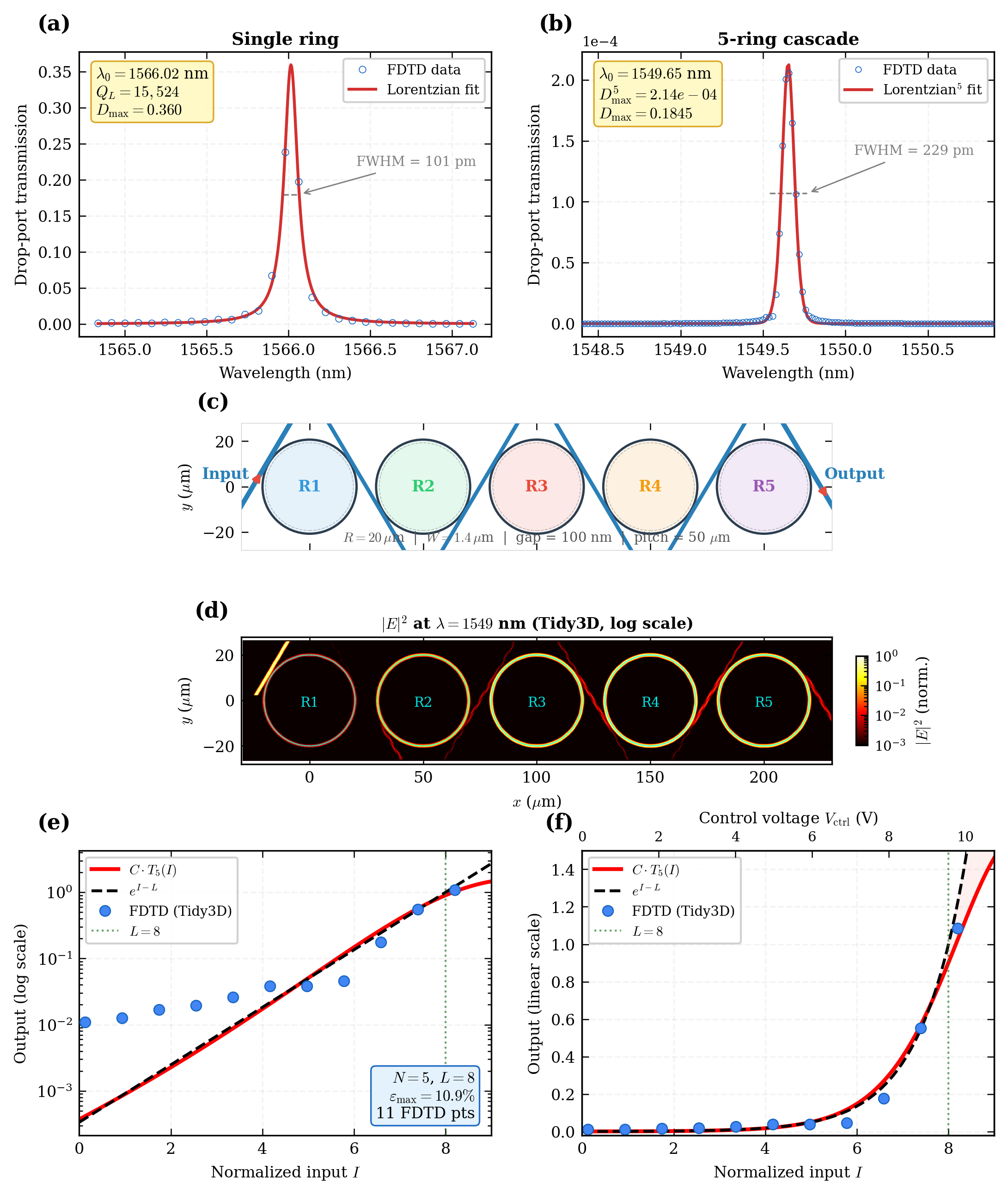}
\caption{FDTD-based AEF validation.
\textbf{(a)}~Lorentzian fit to the drop-port resonance at
$\lambda = \SI{1566}{\nano\meter}$ from 3D FDTD (Lumerical)
($Q_L = 15{,}500$, $D_{\max} = 0.360$, $b_V = \SI{0.180}{\per\volt}$).
\textbf{(b)}~Five-ring cascade drop-port spectrum near $\lambda_0 \approx 1550\;\mathrm{nm}$ with Lorentzian$^5$ fit (red curve), confirming the expected $T^5$ line shape.
\textbf{(c)}~Five-ring cascade MRR layout with diagonal zigzag bus waveguides.
\textbf{(d)}~$|E|^2$ field profile at $\lambda = \SI{1549}{\nano\meter}$
from a five-ring cascade 3D FDTD simulation (Tidy3D~\cite{tidy3d2024}).
\textbf{(e,\,f)}~AEF validation of the five-ring cascade on log~(e) and linear~(f) scales
with 11~spectral FDTD data points.}
\label{fig:fdtd_aef}
\end{figure*}

\subsection{X-cut electrode design and EO parameters}
\label{sec:electrode}

We employ lateral signal--ground (S--G) arc electrodes on the slab
surface alongside the ring waveguide (Fig.~\ref{fig:eo_electrode}).
In the X-cut orientation, the crystal Z-axis is at $45^\circ$ from the
horizontal in the substrate plane, giving a lateral-field projection
proportional to $\cos(\theta - 45^\circ)$ at azimuthal angle~$\theta$.
The $\cos(\theta - 45^\circ) = 0$ boundaries at $\theta = 135^\circ$
and $315^\circ$ naturally separate the coupling regions from the
electrode regions.  Each ring carries a full semicircular arc electrode
on the side opposite to its coupling points, engaging the large
$r_{33} = \SI{30.9}{\pico\meter\per\volt}$ Pockels
coefficient~\cite{zhu2021aop, hu2024integratedeo}.
The effective EO fill factor follows from integrating
$|\cos(\theta - 45^\circ)|$ over the semicircle:
\begin{equation}
  f_\mathrm{EO}
  = \frac{1}{\pi}
  \approx 0.318
  \label{eq:fill_factor}
\end{equation}
(see Supplementary Sec.~S4 for derivation).
The electrode gap is $g_\mathrm{el} = \SI{5}{\micro\meter}$
($d_\mathrm{eff} \approx \SI{2.5}{\micro\meter}$), and the
electro-optic overlap integral is $\Gamma_\mathrm{EO} = 0.7$.
Table~\ref{tab:eo_params} lists the electrode parameters.

\begin{figure}[htbp]
\centering
\includegraphics[width=\columnwidth]{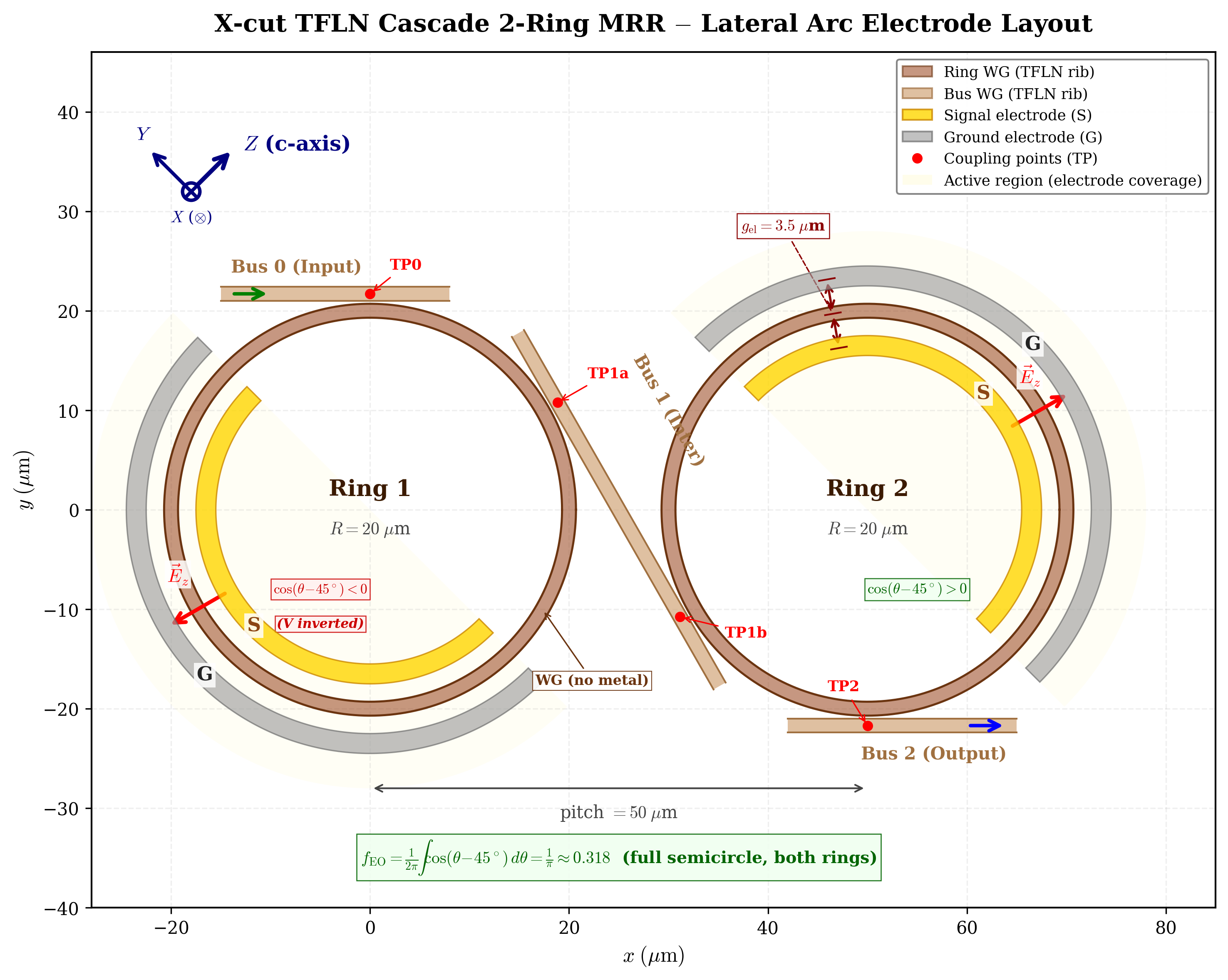}
\caption{Illustrative two-ring cascade layout showing the lateral S--G arc
electrode placement on X-cut TFLN (the cascade design extends to $N$ rings; this two-ring
example clarifies the electrode geometry).  The crystal Z-axis is oriented at $45^\circ$
from the horizontal in the substrate plane.  The
$\cos(\theta-45^\circ)=0$ boundaries at $\theta=135^\circ$ and
$315^\circ$ naturally separate the bus-waveguide coupling regions from
the electrode semicircles: each ring carries a full semicircular arc
electrode on the side opposite to its coupling points.  The resulting
effective EO fill factor is $f_\mathrm{EO} = 1/\pi \approx 0.318$.}
\label{fig:eo_electrode}
\end{figure}

\begin{table}[htbp]
\centering
\caption{Electro-optic electrode parameters for the X-cut TFLN
micro-ring with lateral S--G arc electrodes.}
\label{tab:eo_params}
\begin{ruledtabular}
\begin{tabular}{llcl}
Parameter & Symbol & Value & Unit \\
\hline
Crystal orientation    & ---                      & X-cut & --- \\
EO coefficient         & $r_{33}$                 & 30.9  & \si{\pico\meter\per\volt} \\
EO fill factor         & $f_\mathrm{EO}$          & $1/\pi \approx 0.318$ & --- \\
EO overlap factor      & $\Gamma_\mathrm{EO}$     & 0.7   & --- \\
Electrode gap          & $g_\mathrm{el}$          & 5     & \si{\micro\meter} \\
Effective electrode distance & $d_\mathrm{eff}$   & 2.5   & \si{\micro\meter} \\
\end{tabular}
\end{ruledtabular}
\end{table}

\subsection{FDTD-Calibrated \texorpdfstring{$b_V$}{b\_V} and Cascade Optimization}
\label{sec:cascade_opt}

From the device parameters in Tables~\ref{tab:device_params}
and~\ref{tab:eo_params} and the FDTD-calibrated $n_g \approx 2.30$,
the effective normalized voltage sensitivity is (Supplementary Sec.~S4; here $d\lambda/dV = 28.5\;\mathrm{pm/V}$ is the straight-section value and $f_\mathrm{EO}$ accounts for partial electrode coverage of the ring circumference):
\begin{equation}
  b_V
  = \frac{2\,Q\,(d\lambda/dV)}{\lambda_0}\,f_\mathrm{EO}
  \approx \SI{0.182}{\per\volt}
  \label{eq:bV_device}
\end{equation}
at $Q_L = 15{,}500$.
This estimate relies on a first-order electrostatic model
($d_\mathrm{eff} \approx 2.5\;\mu\mathrm{m}$,
$\Gamma_\mathrm{EO} = 0.7$); a $\pm 30\%$ variation in $b_V$
would shift the cascade depth by one to two rings at constant
$\varepsilon_\mathrm{max}$ (Table~\ref{tab:cascade}), leaving the
qualitative design conclusions unchanged.
With the cascade framework of Sec.~II
(Eqs.~\ref{eq:y}--\ref{eq:Einf}), the $N$-ring drop-port
transmission $\tilde{y}(I) = C\,[1 + (a + bI)^2]^{-N}$
approximates $e^{I-L}$ over $I \in [0, L]$, with $(a, b)$
optimized by minimax fitting for each $N$.

Table~\ref{tab:cascade} presents the optimization results for the
standard dynamic range $L = 8$ ($e^8 \approx 2981$,
\SI{34.7}{\decibel}).

\begin{table}[htbp]
\centering
\caption{Cascade optimization results for $L = 8$.  The bias voltage
$V_\mathrm{bias} = |a|/b_V$ sets the DC offset, and
$V_\mathrm{ctrl} = bL/b_V$ is the maximum control voltage at
$I = L$.  Voltages computed with
$b_V = \SI{0.182}{\per\volt}$ (X-cut arc electrode,
FDTD-calibrated best resonance $Q_L = 15{,}500$,
$n_g = 2.30$).  The mean FDTD quality factor across five resonances
is $Q_L = 12{,}500 \pm 1{,}800$; using the mean would increase
voltages by ${\sim}24\%$.}
\label{tab:cascade}
\begin{ruledtabular}
\begin{tabular}{rcccccc}
$N$ & $a$ & $b$ & $E_\infty$ & $\varepsilon_\mathrm{max}$ (\%)
    & $V_\mathrm{bias}$ (\si{\volt}) & $V_\mathrm{ctrl}$ (\si{\volt}) \\
\hline
 5  & $-2.0789$ & 0.21658 & 0.1035 & 10.91 & 11.4 &  9.5 \\
10  & $-1.4588$ & 0.10202 & 0.0265 &  2.68 &  8.0 &  4.5 \\
12  & $-1.3731$ & 0.08450 & 0.0184 &  1.86 &  7.5 &  3.7 \\
20  & $-1.2136$ & 0.05025 & 0.0067 &  0.67 &  6.7 &  2.2 \\
25  & $-1.1685$ & 0.04013 & 0.0043 &  0.43 &  6.4 &  1.8 \\
30  & $-1.141$ & 0.03340 & 0.0030 & 0.30 & 6.3 & 1.5 \\
32  & $-1.1301$ & 0.03131 & 0.0026 &  0.26 &  6.2 &  1.4 \\
\end{tabular}
\end{ruledtabular}
\footnotetext{The complete cascade optimization results for all $N$ values are listed in Supplementary Table~\ref{tab:S-cascade}.}
\end{table}

The approximation quality across different cascade depths is shown in
Fig.~\ref{fig:cascade_fit} (Sec.~III).
Key thresholds (e.g., $\varepsilon < 2\%$ at $N \geq 12$,
$\varepsilon < 1\%$ at $N \geq 17$) and the complete optimization
results are listed in Supplementary Sec.~S4.



\section{Physical Feasibility}
\label{sec:feasibility}

Having established the cascade approximation theory (Sec.~II)
and the FDTD-calibrated device parameters (Sec.~\ref{sec:device}),
we now assess the physical feasibility of the proposed architecture
in terms of voltage requirements, insertion loss, and energy efficiency.

\subsection{Electro-optic voltage requirements}
\label{sec:voltage}

For the primary target of $\varepsilon < 2\%$ ($N = 12$), minimax optimization gives $a = -1.373$, $b = 0.0845$.
With the FDTD-calibrated $Q_L = 15{,}500$ ($b_V = 0.182\;\mathrm{V}^{-1}$), the required voltages are
\begin{align}
  V_{\mathrm{bias}} &= \frac{|a|}{b_V}
    = \frac{1.373}{0.182} = 7.5\;\mathrm{V},
  \label{eq:Vbias} \\
  V_{\mathrm{ctrl,max}} &= \frac{bL}{b_V}
    = \frac{0.0845 \times 8}{0.182} = 3.7\;\mathrm{V}.
  \label{eq:Vctrl}
\end{align}
Since $b_V \propto Q$, voltage scales inversely with quality factor:
\begin{equation}
  V_{\mathrm{ctrl}} = \frac{bL}{b_V}
    = \frac{bL\,\lambda_0}{2Q\,|d\lambda_0/dV|}.
  \label{eq:Vscaling_feas}
\end{equation}
CMOS-compatible control voltages ($V_{\mathrm{ctrl}} < 3.3\;\mathrm{V}$) are achievable at $N \geq 14$ with $Q_L = 15{,}500$; at the design point $N = 30$ ($\varepsilon_\mathrm{max} = 0.30\%$), $V_\mathrm{ctrl} = 1.47\;\mathrm{V}$.

\subsection{Power budget: two-regime analysis}
\label{sec:power_budget}

The on-resonance cascade transmission $D_\mathrm{max}^N$ is the dominant contribution to total insertion loss.
Table~\ref{tab:power_budget} presents two regimes: the FDTD-characterized regime ($D_\mathrm{max} = 0.36$) and the fabricated high-$Q$ regime ($D_\mathrm{max} = 0.95$, achievable with $Q_i > 10^6$ and gap-optimized coupling).

\begin{table}[t]
\centering
\caption{Two-regime power budget for the MRR cascade.
$P_\mathrm{out}$ assumes per-channel input
$P_\mathrm{in,ch} = 100\;\mu\mathrm{W}$
(from a shared $P_\mathrm{in,tot} = 1\;\mathrm{mW}$ CW laser split
across $M = 10$ parallel channels via a $1\!\times\!M$ splitter,
or equivalently multiplexed as $d$ WDM channels sharing a single cascade) and accounts only for the ideal
on-resonance cascade transmission $D_\mathrm{max}^N$ (upper bound);
additional inter-ring coupling loss
($\eta_\mathrm{coupling} \approx 0.9$ per stage,
${\sim}0.46\;\mathrm{dB}$/stage)
and off-resonance propagation loss
($0.08$--$0.25\;\mathrm{dB}$/stage) are analyzed
separately in Sec.~\ref{sec:loss_budget}.}
\label{tab:power_budget}
\begin{ruledtabular}
\begin{tabular}{cccccccc}
 & $D_\mathrm{max}$ & $N$ & $D_\mathrm{max}^N$ & (dB)
       & $P_\mathrm{out}$ & $\varepsilon_\mathrm{max}$ \\
\hline
\multirow{3}{*}{\shortstack{I\\(FDTD)}}
  & 0.36 &  3 & 0.0467  & $-13.3$ & $4.67\;\mu\mathrm{W}$  & ${\sim}15\%$ \\
  & 0.36 &  5 & 0.00605 & $-22.2$ & $0.61\;\mu\mathrm{W}$  & $10.9\%$     \\
  & 0.36 &  7 & $7.84 \times 10^{-4}$ & $-31.1$ & $78\;\mathrm{nW}$  & ${\sim}5\%$  \\
\hline
\multirow{3}{*}{\shortstack{II\\(high-$Q$)}}
  & 0.95 & 10 & 0.599   & $-2.2$  & $59.9\;\mu\mathrm{W}$  & $2.68\%$     \\
  & 0.95 & 20 & 0.358   & $-4.5$  & $35.8\;\mu\mathrm{W}$  & $0.67\%$     \\
  & 0.95 & 30 & 0.215   & $-6.7$  & $21.5\;\mu\mathrm{W}$  & ${\sim}0.30\%$ \\
\end{tabular}
\end{ruledtabular}
\vspace{-2pt}
{\footnotesize Regime~I: FDTD-characterized ($Q_i = 38{,}800$).
Regime~II: fabricated high-$Q$ ($Q_i > 10^6$).
$P_\mathrm{out}$ scales linearly with $P_\mathrm{in,ch}$.}
\end{table}

In the FDTD-characterized regime, $D_\mathrm{max} = 0.36$ limits practical cascades to $N \leq 5$: at $N = 5$ the output is $0.61\;\mu\mathrm{W}$ ($-22.2\;\mathrm{dB}$) with $\varepsilon = 10.9\%$, suited for proof-of-concept validation.
In the fabricated high-$Q$ regime ($D_\mathrm{max} \geq 0.95$), deep cascades become practical: $N = 30$ yields $P_\mathrm{out} = 21.5\;\mu\mathrm{W}$ ($-6.7\;\mathrm{dB}$) with $\varepsilon_\mathrm{max} \approx 0.30\%$.
The transition to fabricated high-$Q$ devices is therefore critical for achieving both high accuracy and sufficient output power.

\subsection{Feasibility outlook}
\label{sec:highQ_projections}

Published TFLN micro-ring resonators achieve $Q_i \ge 10^6$--$10^8$
using optimized fabrication~\cite{zhang2017highQ_optica,gao2022ultrahighQ,zhuang2023wetetch,zhu2024highQ_racetrack}.
At $Q_i = 10^6$ with the present coupling geometry, CMT predicts
$D_\mathrm{max} \approx 0.95$ and $Q_L \approx 25{,}200$
(Supplementary Sec.~S5, Tables~S4--S7),
enabling deep cascades ($N \le 30$) with sub-percent error.
The literature values provide strong independent evidence that intrinsic quality factors in the projected range are physically achievable in TFLN---albeit with wider waveguides and larger ring radii than the present design.
Transferring comparable sidewall quality to our geometry ($R = \SI{20}{\micro\meter}$, $W = \SI{1.4}{\micro\meter}$) is an open fabrication challenge; the projections should be read as design targets contingent on achieving it.

\label{sec:loss_budget}
The total insertion loss comprises on-resonance cascade transmission $D_\mathrm{max}^N$, inter-ring coupling loss (${\sim}0.46\;\mathrm{dB}$/stage for the present diagonal-bus layout), off-resonance propagation loss ($0.08$--$0.25\;\mathrm{dB}$/stage), and fiber-to-chip coupling ($1.5$--$3.0\;\mathrm{dB}$).
For the fabricated high-$Q$ regime ($N = 30$), the total ranges from ${\sim}13\;\mathrm{dB}$ (optimized layout) to ${\sim}24\;\mathrm{dB}$ (current geometry); see Supplementary Sec.~S6 for detailed scenarios.

\subsection{Energy comparison}
\label{sec:energy}

For $N = 30$ X-cut TFLN micro-ring resonators in the fabricated
high-$Q$ regime ($Q_L \approx 25{,}200$ at $Q_i = 10^6$;
Supplementary Sec.~S5), the three energy components are EO tuning
($E_\mathrm{EO} = 0.22$~pJ), amortized laser
($E_\mathrm{laser} = 0.07$~pJ, shared across $M = 10$ channels),
and photodetector ($E_\mathrm{PD} = 0.50$~pJ), yielding
$E_\mathrm{photonic} = 0.79$~pJ
(derivations in Supplementary Sec.~S7).
Including thermal stabilization for $N = 30$ rings
($0.15$--$0.60\;\mathrm{pJ}$; Supplementary Sec.~S7),
the total rises to $0.94$--$1.39\;\mathrm{pJ}$.

Table~\ref{tab:exp_comparison} compares the photonic cascade with
digital implementations.
Including thermal stabilization ($0.94$--$1.39\;\mathrm{pJ}$), the
advantage over INT8 ($2.3\;\mathrm{pJ}$) is $1.7$--$2.4\times$,
while operating at $10\;\mathrm{GHz}$ bandwidth and
$58\times$ lower than digital FP32 ($46\;\mathrm{pJ}$).
At fabricated $Q \ge 30{,}000$, $E_\mathrm{EO}$ drops to
$0.16\;\mathrm{pJ}$ and $E_\mathrm{total} \approx 0.73\;\mathrm{pJ}$
(excluding thermal; Supplementary Table~\ref{tab:energy_vs_Q}),
recovering a $3.2\times$ advantage over INT8.
Since $E_\mathrm{EO} \propto 1/Q^2$, improving $Q$ beyond
${\sim}30{,}000$ yields diminishing energy returns but continues to
relax CMOS driver voltage requirements.

\begin{table}[t]
\centering
\caption{Energy per exponential operation: single-channel comparison.}
\label{tab:exp_comparison}
\begin{ruledtabular}
\begin{tabular}{lccc}
Implementation & $E/\mathrm{op}$ (pJ) & Bandwidth & Notes \\
\hline
Digital FP32 (Taylor)                & ${\sim}46$  & 1\,GHz   & 10 FP MACs \\
Digital INT8 (Taylor)                & ${\sim}2.3$ & 1\,GHz   & 10 INT MACs \\
\textbf{Photonic MRR ($N\!=\!30$)}   & $\mathbf{0.94}$--$\mathbf{1.39}$ & \textbf{10\,GHz} & \textbf{Analog$^\dagger$} \\
\end{tabular}
\end{ruledtabular}
{\footnotesize $^\dagger$0.79~pJ excluding thermal; 0.94--1.39~pJ including thermal. Self-consistent with fabricated high-$Q$ regime ($Q_L = 25{,}200$); see Supplementary Sec.~S7.}
\end{table}



\section{Discussion}
\textbf{Practical design procedure.}
For a given input sequence $x=(x_1,\dots,x_K)$, the design proceeds as follows:
\begin{enumerate}
\item Compute $m=\max_n x_n$, $u_n=x_n-m$, and $L=-\min_n u_n$.
\item Map to nonnegative control-signal amplitudes: $I_n=u_n+L \in[0,L]$.
\item Choose tolerance $\varepsilon$ and set $\varepsilon_{\log}=\ln(1+\varepsilon)$.
\item Select a physically feasible $b_{\max}$ and estimate $N$ using Eq.~(\ref{eq:Nestimate}).
\item Initialize $b=\min(b_{\max},1/N)$ and $a=-1-bL/2$, then refine $(a,b)$ by a two-parameter minimax fit if required.
\item The optical block yields $\tilde y(I_n)\approx e^{x_n-m}$, and softmax weights follow as

\begin{equation}
p_n=\frac{\tilde y(I_n)}{\sum_j \tilde y(I_j)}.
\end{equation}
\end{enumerate}

\textbf{Scope and limits.}
The approximation is for a finite interval $I\in[0,L]$, where $L$ is determined by the input batch via Eq.~(\ref{eq:defL}).
In practice, one designs for a worst-case $L$ expected in operation (or retunes $a$ and rescales the control signal to adapt $L$).
Noise, insertion loss, and control-induced parasitics limit accuracy and dynamic range; we treat these effects as platform-specific margins.
Detailed non-ideality assumptions, parameter distributions, and robustness statistics are reported in Supplementary Sec.~S8.
With $K$ channels in parallel, one can form softmax by summing channel powers and applying a shared reciprocal scale factor, depending on the chosen mixed-signal normalization scheme.

\textbf{WDM parallelism.}
A particularly hardware-efficient realization exploits
wavelength-division multiplexing (WDM):
$d$ probe wavelengths $\lambda_1,\dots,\lambda_d$, each
resonant with a distinct FSR order of the same ring set, traverse
a \emph{single} $N$-ring cascade simultaneously
(Fig.~\ref{fig:wdm_layout}).
Because each channel $\lambda_j$ sees its own Lorentzian
detuning set by an independent control voltage $V_j$, the cascade
output per channel is
$\tilde y_j = C\prod_{k=1}^{N}T_{\mathrm{drop},k}(\lambda_j,V_j)
\approx e^{V_j}$,
and all $d$ exponentials are computed in parallel on the same
physical waveguide.
Compared with a $1\!\times\!M$ power-splitter architecture that
replicates the cascade for each channel, the WDM approach reduces
the total ring count from $N \times d$ to $N$ (a factor-$d$
saving) and eliminates the splitter insertion loss
($10\log_{10} d$~dB).
At the output, a WDM demultiplexer or wavelength-selective
photodetector array separates the channels for electrical readout.
Figure~\ref{fig:wdm_layout} shows a representative chip layout
for $N = 5$ cascade stages and $d = 8$ WDM channels, where
alternating U-turn bus connections route the drop-port output
of each stage into the input bus of the next.

\textbf{Why cascade helps.}
A single Lorentzian in $I$ is too rigid to mimic the log-linear target over a wide interval.
Cascading turns the transfer into a product; taking a logarithm gives a sum of smooth terms, and the approximation improves as $N$ increases.
The slope constraint $N|b|\gtrsim 1$ is an immediate feasibility check.

\textbf{Global softmax normalization via WDM feedback.}
The WDM-parallel architecture (Fig.~\ref{fig:wdm_layout}) integrates
naturally with a closed-loop normalization scheme to complete the full
softmax function.
After the $N$-stage cascade, a WDM demultiplexer (e.g., arrayed-waveguide
grating or ring-filter bank) routes each channel $\lambda_j$ to a
dedicated photodetector, producing photocurrents
$I_{\lambda_j} \propto \tilde y_j \approx C\,P_\mathrm{in}\,e^{V_j}$.
The $d$ photocurrents are summed electrically:
\begin{equation}
  S = \sum_{j=1}^{d} I_{\lambda_j}
    \propto C\,P_\mathrm{in}\sum_{j=1}^{d} e^{V_j}.
  \label{eq:wdm_sum}
\end{equation}
A proportional--integral (PI) controller compares $S$ with a fixed
reference $S_\mathrm{ref}$ and adjusts the shared WDM laser power
$P_\mathrm{in}$ so that $S \to S_\mathrm{ref}$~\cite{doylend2010oe,astrom2008feedback}.
Because all $d$ channels share the same probe source, scaling
$P_\mathrm{in}$ multiplies every $\tilde y_j$ by the same factor;
upon convergence
\begin{equation}
  p_j = \frac{I_{\lambda_j}}{S_\mathrm{ref}}
      = \frac{e^{V_j}}{\sum_{k=1}^{d}e^{V_k}}
      = \mathrm{softmax}(V)_j,
  \label{eq:wdm_softmax}
\end{equation}
realizing the complete softmax with a single feedback loop and no
per-channel normalization circuitry.
Compared with the replicated-cascade approach (one AEF block per
channel), WDM feedback offers two additional benefits:
(i)~the splitter-induced power imbalance that would bias the
$I_{\lambda_j}$ ratios is absent, since all channels traverse the same optical
path; and
(ii)~a single laser control point replaces $d$ independent probe
adjustments.
Design details and stability analysis of the PI loop are provided
in Supplementary Sec.~S9.

Beyond ring-resonator AEF implementations, the same cascade principle can be extended to other cavity-based photonic platforms, such as serial 1D photonic-crystal cavities and other cascaded resonant architectures~\cite{heebner2008optical,yariv1999crow}.
What these platforms share is transfer-function shaping through cascaded resonances; loss, tuning range, fabrication tolerance, and calibration overhead remain platform-dependent.

\begin{figure*}[t]
\centering
\includegraphics[width=\textwidth]{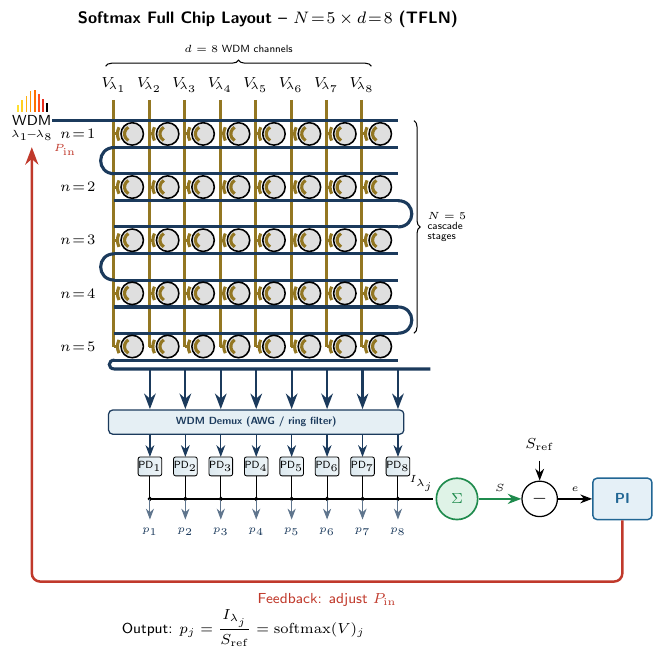}
\caption{WDM-parallel MRR-AEF system with closed-loop softmax normalization
($N\!=\!5$ cascade stages, $d\!=\!8$ WDM channels) on X-cut TFLN.
A single WDM source ($\lambda_1$--$\lambda_8$) enters the top
input bus waveguide; each stage applies a Lorentzian drop-port
transfer, and alternating U-turn connections route the drop-port output
into the next stage's input bus.
Per-channel EO bias voltages ($V_{\lambda_1}$--$V_{\lambda_8}$)
independently tune each column of rings.
The final drop output passes through a WDM demultiplexer (AWG / ring filter)
and is detected by a PD array, producing per-channel photocurrents
$I_{\lambda_j} \propto e^{V_j}$.
The photocurrents are summed ($\Sigma$) and compared with a reference
$S_\mathrm{ref}$; a PI controller adjusts the shared laser power $P_\mathrm{in}$
until $S = S_\mathrm{ref}$, at which point each PD output directly yields
$p_j = I_{\lambda_j}/S_\mathrm{ref} = \mathrm{softmax}(V)_j$
(Eq.~\ref{eq:wdm_softmax}).}
\label{fig:wdm_layout}
\end{figure*}


The insertion loss budget (Sec.~\ref{sec:loss_budget}) and electro-optic voltage requirements (Sec.~\ref{sec:voltage}) suggest that the cascade architecture is feasible under optimized coupling and layout conditions.
Using monolithic TFLN microring data from Bahadori \textit{et al.}~\cite{bahadori2020oe} ($Q \approx 5432$, $d\lambda_0/dV \approx 9$--$20$~pm/V), the normalized sensitivity $b_V \simeq 0.063$--$0.14$~V$^{-1}$, within the range required by the cascade design.

\noindent\textbf{Crystal orientation and electrode design.}
The X-cut TFLN platform was chosen for several reasons.
First, X-cut is the prevailing industry standard for integrated TFLN
modulators, with well-established fabrication processes and commercial
wafer availability~\cite{zhu2021aop,hu2024integratedeo}.
Second, the TE$_0$ mode---which is strongly confined in the rib
waveguide geometry---can engage the large $r_{33}$ coefficient
via lateral electric fields aligned with the crystal Z-axis.
In contrast, Z-cut geometry with TE polarization can only access the
smaller $r_{13}$ coefficient ($\sim 10\;\mathrm{pm/V}$), resulting in
significantly lower electro-optic efficiency.
The arc electrode design (Sec.~\ref{sec:electrode}) addresses the
phase-cancellation problem inherent to X-cut circular
rings~\cite{bahadori2020oe} by orienting the crystal Z-axis at
$45^\circ$ from the horizontal in the substrate plane.  This rotation
places the $\cos(\theta - 45^\circ) = 0$ boundaries at
$\theta = 135^\circ$ and $315^\circ$, naturally separating the
bus-waveguide coupling regions from the electrode regions.  Each ring
carries a full semicircular arc electrode on the side opposite to its
coupling points, yielding an effective fill factor
$f_\mathrm{EO} = 1/\pi \approx 0.318$.
While this reduces the round-trip EO efficiency compared to a
hypothetical full-circumference design, it preserves the compact
footprint of a circular ring resonator.
The cascade performance can be further improved beyond the
$R = 20\;\mu\mathrm{m}$ circular-ring design presented here.
Increasing the ring radius reduces bending loss and raises the
intrinsic quality factor $Q_i$, which directly increases $b_V$
($\propto Q$) and lowers the required control voltage.
Alternatively, adopting a racetrack geometry with extended straight
coupling sections strengthens the bus--ring coupling, pushing the
drop-port maximum $D_{\max}$ closer to critical coupling and
improving the per-stage transfer efficiency.
Either approach---or their combination---would yield higher $b_V$
and $D_{\max}$, enabling lower $N$ or tighter approximation accuracy
at reduced operating voltages.

\noindent\textbf{Fabrication considerations.}
The X-cut TFLN rib waveguide (600~nm total thickness, 500~nm etch, $w=1.4$~$\mu$m) follows established fabrication processes for commercial TFLN wafers on SiO$_2$~\cite{zhu2021aop,hu2024integratedeo}.
The lateral signal--ground (SG) electrode configuration is fabricated in a single metal layer, which is standard in TFLN foundry processes.
The primary fabrication challenge for the cascade architecture is maintaining uniform coupling gaps ($g=100$~nm) across $N$ rings to ensure identical Lorentzian transfer functions.
Post-fabrication trimming via UV exposure or localized thermal oxidation can compensate residual detuning variations~\cite{resonancetrim2021}, as quantified in the Monte Carlo robustness analysis (Supplementary Sec.~S8).
Monte Carlo simulations (Supplementary Sec.~S8) show that under nominal non-ideality levels ($\sigma_a = 0.020$, $\sigma_{b,\mathrm{rel}} = 0.020$), a single-point calibration of $C$ per chip keeps the median softmax KL divergence below $2.2 \times 10^{-4}$, with 95th-percentile max probability error under $0.32\%$. Even under stress conditions ($\sigma_a = 0.032$), 95th-percentile errors remain below $0.42\%$, demonstrating that the identical-detuning design is robust to realistic fabrication variations provided a per-chip calibration step is performed.
Conversely, if coupling gaps are intentionally varied across rings, the per-ring parameters $(a_k, b_k)$ become independent degrees of freedom.
A Taylor-expansion analysis shows that $K$ non-identical rings can cancel curvature terms up to order $2K$ in the Taylor series of $g(I)=\sum_k\ln T_k$, one order higher than $K$ identical rings, so that fewer rings suffice for a given error target.

\begin{table}[t]
\centering
\caption{Summary of evidence levels for key claims.}
\label{tab:evidence_levels}
\begin{ruledtabular}
\begin{tabular}{lcc}
Claim & Evidence & Sec.\ \\
\hline
Cascade $\to$ exp.\ approx.  & Analytic       & II \\
Depth scaling                & Analytic + num.\ & II, III \\
$Q_L$, $D_\mathrm{max}$, $b_V$ & 3-D FDTD    & IV \\
5-ring line shape            & 3-D FDTD       & IV \\
$N\!\le\!30$ deep cascade   & CMT proj.$^*$  & V \\
Energy ${<}\,1$~pJ           & Estimate       & V \\
Full softmax (WDM + feedback) & Conceptual + layout & VI \\
\end{tabular}
\end{ruledtabular}
{\footnotesize $^*$Based on published $Q_i \ge 10^6$ values~\cite{zhang2017highQ_optica,gao2022ultrahighQ} and CMT coupling model.}
\end{table}

\section{Conclusion}
We have presented a cascaded micro-ring resonator architecture that approximates the exponential function $e^{x_n - m}$ on a finite interval $[0,L]$ using multiplicative Lorentzian transfer functions.
Increasing the cascade depth $N$ systematically reduces the worst-case relative error, and an identical-detuning design initialized by flank and slope matching provides a practical two-parameter design.

Three-dimensional FDTD simulations of a single X-cut TFLN add-drop ring ($R=20$~$\mu$m, $g=100$~nm) yield $Q_L = 10{,}300\text{--}15{,}500$ and $D_{\max}\approx 0.36$, calibrating the cascade transfer model.
A five-ring cascade 3D FDTD simulation directly validates the multi-ring framework: all five rings exhibit resonant excitation, and mapping the drop-port spectrum onto the dimensionless control variable reproduces the theoretical $N = 5$ curve with ${\sim}11\%$ integrated relative-area error over the upper operating range ($I \ge 5.8$), providing the first multi-ring confirmation of the cascade exponential approximation.
At the present FDTD-characterized quality factor, practical cascades are limited to $N = 5$--$7$ ($\varepsilon\lesssim 11\%$).
If high-$Q$ TFLN resonators reported in the literature ($Q_i\ge 10^6$, $D_{\max}\ge 0.95$) are realized in the cascade geometry, deeper cascades ($N \approx 20$--$30$) would reach sub-percent approximation error with an estimated per-operation energy of $0.79$--$1.39$~pJ, which is $1.7$--$2.4\times$ lower than an INT8 MAC at the 7~nm node.
Monte Carlo analysis shows that the identical-detuning design tolerates realistic fabrication variations ($\sigma_a = 0.020$, $\sigma_{b,\mathrm{rel}} = 0.020$) with a single per-chip calibration, keeping the 95th-percentile softmax probability error below $0.32\%$.

The formulation is not restricted to electro-optic tuning: it requires only a controllable detuning coordinate with local linearization, so both Pockels and optical (Kerr/XPM) mechanisms are compatible~\cite{zhu2021aop,ahmed2019ol,bahadori2020oe,hu2024integratedeo}.
We demonstrate a photonic exponential block and present a WDM-parallel chip architecture (Fig.~\ref{fig:wdm_layout}) in which $d$ wavelength channels share a single $N$-ring cascade, reducing the total ring count by a factor of $d$ and eliminating power-splitter loss.
Combined with a single-loop PI feedback that adjusts the shared WDM laser power, the architecture realizes the complete softmax function---exponentiation, summation, and normalization---without per-channel normalization circuitry.
Max-finding and digital interfacing remain open for future experimental validation.

\FloatBarrier
\putbib[refs_new]
\end{bibunit}

\clearpage
\onecolumngrid
\begin{bibunit}

\section*{Supplementary Information}
\noindent\textit{Supplementary material accompanying ``Photonic Exponential Approximation via Cascaded TFLN Microring Resonators toward Softmax.''}
\setcounter{figure}{0}
\setcounter{table}{0}
\renewcommand{\thefigure}{S\arabic{figure}}
\renewcommand{\thetable}{S\arabic{table}}
\renewcommand{\theHfigure}{suppfig.\arabic{figure}}
\renewcommand{\theHtable}{supptab.\arabic{table}}

\section*{S0. Rigorous derivation and validity scope}

This section derives the depth-scaling relations and screening bounds used in the main text,
and states the assumptions under which they apply, together with validity scope and failure cases.
We separate proved statements (Lemma, Proposition, Theorem) from heuristic engineering estimates that rely on empirical calibration.

\subsection*{S0.1~~Assumptions}

\begin{assumption}[Lorentzian single-ring transfer]\label{A1}
Each ring $k$ has a normalized add-to-drop transmission of the form
$T_k(I)=[1+(a_k+bI)^2]^{-1}$, where $a_k\in\mathbb{R}$ is the
dimensionless static detuning, $b>0$ is the common normalized
sensitivity, and $I\ge 0$ is a nonnegative control-signal amplitude.
\end{assumption}

\begin{assumption}[Multiplicative cascade]\label{A2}
The $N$ rings are cascaded in a serial add-drop topology (the drop
output of ring~$k$ feeds the add input of ring~$k\!+\!1$), and the probe
is sufficiently weak that cross-ring and nonlinear probe-induced
effects are negligible; thus the total normalized transmission is
$y(I)=\prod_{k=1}^{N}T_k(I)$.
\end{assumption}

\begin{assumption}[Identical-detuning family]\label{A3}
All rings share the same static detuning: $a_1=\cdots=a_N\equiv a$.
This reduces the design space to $(a,b)$ and a global scale $C>0$;
the scaled output is $\tilde{y}(I)=C\,[1+(a+bI)^{2}]^{-N}$.
\end{assumption}

\begin{assumption}[Linear control-to-resonance mapping]\label{A4}
Within the operating range $I\in[0,L]$, the resonance shift is a
linear function of the control-signal amplitude (Eq.~(10) of main text),
i.e., higher-order detuning nonlinearity is negligible.
\end{assumption}

\begin{assumption}[Finite interval and bounded $L$]\label{A5}
The approximation target is $f(I)=e^{I-L}$ on a finite interval
$I\in[0,L]$ with $L>0$ determined by the input batch
($L=\max(x)-\min(x)$). The depth-scaling results are derived for
fixed, finite $L$.
\end{assumption}

\begin{assumption}[Flank-centered operating regime]\label{A6}
The design uses the ``flank-centered'' initialization:
$a+b(L/2)=-1$ (midpoint on the Lorentzian half-maximum) and
$Nb=1$ (slope matching). This places the operating point in the
steepest-slope region of the Lorentzian, where the log-transfer is
most nearly linear.
\end{assumption}

\subsection*{S0.2~~Rigorous results}

Throughout, define the log-domain residual
\begin{equation}
r(I)\;\equiv\;\ln\tilde{y}(I)-(I-L)
\;=\;\ln C - N\ln\!\bigl(1+(a+bI)^{2}\bigr)-(I-L),
\tag{S0.1}
\end{equation}
and the worst-case log-error
$E_\infty=\sup_{I\in[0,L]}|r(I)|$.
We set $C$ to the minimax-optimal value
$\ln C^{\star}=-\bigl(\max_{I}g(I)+\min_{I}g(I)\bigr)/2$,
where $g(I)=\ln y(I)-(I-L)$, throughout.

\begin{lemma}[Slope bound --- rigorous]\label{lem:slope}
Under Assumptions~\ref{A1}--\ref{A3}, for every $I\ge 0$,
\[
\left|\frac{d}{dI}\ln y(I)\right| \;\le\; N\,|b|.
\]
\end{lemma}
\begin{proof}
From Assumption~\ref{A3},
$\ln y(I)=-N\ln\!\bigl(1+(a+bI)^{2}\bigr)$.
Differentiating:
\[
\frac{d}{dI}\ln y(I)
\;=\;-N\,\frac{2b(a+bI)}{1+(a+bI)^{2}}.
\]
Let $u\equiv a+bI\in\mathbb{R}$.
The function $h(u)=2u/(1+u^{2})$ satisfies
$|h(u)|\le 1$ for all $u\in\mathbb{R}$
(since $1+u^{2}\ge 2|u|$ by AM--GM).
Therefore $|d(\ln y)/dI|=N|b|\,|h(u)|\le N|b|$.
\end{proof}

\begin{remark}[Necessary condition for approximation]
Since the target $\ln f(I)=I-L$ has constant slope $+1$, a necessary
condition for the cascade log-transfer to track this slope at any point
is $N|b|\ge 1$.
This is Eq.~(\ref{eq:Nb}) of the main text and is a rigorous (not heuristic)
necessary condition.
\end{remark}

\begin{proposition}[Log-domain Taylor expansion at flank center]\label{prop:taylor}
Under Assumptions~\ref{A1}--\ref{A6}, define $I_0=L/2$ and
$\delta=I-I_0$. Then
\begin{equation}
\ln\tilde{y}(I)
\;=\;\mathrm{const}+\delta + \frac{\delta^{3}}{6N^{2}}
+ R_4(\delta),
\tag{S0.2}
\end{equation}
where $|R_4(\delta)|\le M_4\,\delta^{4}$ with
$M_4=N|b|^{4}\cdot\sup_{|\delta|\le L/2}|q^{(4)}(u(\delta))|$
and $q(u)=-\ln(1+u^2)$.
In particular, the quadratic term vanishes identically at the
flank point $u_0=a+bI_0=-1$.
\end{proposition}
\begin{proof}
Set $u(\delta)=a+b(I_0+\delta)=-1+b\delta$ (using $a+bI_0=-1$).
Define $\phi(u)=-\ln(1+u^{2})$. Then $\ln y(I)=N\,\phi(u(\delta))$
and $u'(\delta)=b$. Compute derivatives of $\phi$ at $u_0=-1$:
\begin{align*}
\phi'(u)&=-\frac{2u}{1+u^{2}},
&\phi'(-1)&=1,\\
\phi''(u)&=\frac{2(u^{2}-1)}{(1+u^{2})^{2}},
&\phi''(-1)&=0,\\
\phi'''(u)&=\frac{4u(3-u^{2})}{(1+u^{2})^{3}},
&\phi'''(-1)&=\frac{-4(-1)(3-1)}{(1+1)^{3}}=1.
\end{align*}
By the chain rule, writing $F(\delta)=N\phi(u(\delta))$:
\begin{align*}
F'(0)&=Nb\,\phi'(-1)=Nb=1,\\
F''(0)&=Nb^{2}\,\phi''(-1)=0,\\
F'''(0)&=Nb^{3}\,\phi'''(-1)=Nb^{3}=\frac{1}{N^{2}},
\end{align*}
where we used $b=1/N$ from Assumption~\ref{A6} in the last step.
Hence the Taylor expansion with the minimax-optimal $C$ is
\[
\ln\tilde{y}(I)
=\mathrm{const}+\delta+0\cdot\frac{\delta^{2}}{2}
+\frac{1}{N^{2}}\cdot\frac{\delta^{3}}{6}+R_4(\delta).
\]
Subtracting the target $\delta$ (the linear part of $I-L$ around $I_0$)
gives a leading residual $\delta^{3}/(6N^{2})$.
The remainder is bounded by the standard Taylor remainder estimate.
\end{proof}

\begin{theorem}[Heuristic depth-scaling law]\label{thm:scaling}
Under Assumptions~\ref{A1}--\ref{A6} and ignoring the fourth-order
remainder $R_4$, the leading-order worst-case log-error on
$I\in[0,L]$ satisfies
\begin{equation}
E_\infty^{(\mathrm{leading})}
\;\sim\;\frac{1}{6N^{2}}\left(\frac{L}{2}\right)^{3}
\;=\;\frac{L^{3}}{48\,N^{2}}.
\tag{S0.3}
\end{equation}
Setting $E_\infty^{(\mathrm{leading})}\le\varepsilon_{\log}
=\ln(1+\varepsilon)$ and solving for $N$ gives
\begin{equation}
N\;\ge\;\frac{L^{3/2}}{\sqrt{48\,\varepsilon_{\log}}}.
\tag{S0.4}
\end{equation}
\end{theorem}
\begin{proof}[Derivation (heuristic)]
From Proposition~\ref{prop:taylor}, the residual with respect to the
target is dominated by $\delta^{3}/(6N^{2})$ for $|\delta|\le L/2$.
The maximum of $|\delta|^{3}$ on $[-L/2,\,L/2]$ is $(L/2)^{3}$.
Setting the bound equal to $\varepsilon_{\log}$ and solving:
\[
\frac{L^{3}}{48\,N^{2}}\le\varepsilon_{\log}
\quad\Longrightarrow\quad
N\ge \frac{L^{3/2}}{\sqrt{48\,\varepsilon_{\log}}}.
\]
With $1/\sqrt{48}\approx 0.144$, and accounting for the fact that
the minimax-optimal residual is typically smaller than the one-sided
Taylor bound by a factor of $\sim 2$ (equi-oscillation), the
effective prefactor becomes $\kappa\approx 0.07$, yielding the
main-text engineering estimate Eq.~(\ref{eq:Nestimate}):
$N\approx\lceil\max(1/b_{\max},\;\kappa\,L^{3/2}/\sqrt{\varepsilon_{\log}})\rceil$.
\end{proof}

\begin{remark}[Status of Theorem~\ref{thm:scaling}]
\textbf{This is a heuristic scaling law, not a rigorous minimax guarantee.}
The derivation truncates the Taylor series at third order and approximates
the equi-oscillation factor empirically ($\kappa\approx 0.07$).
For a rigorous bound one would need explicit control of $R_4$ over the
full interval $[0,L]$, which depends on $L$, $N$, and higher derivatives
of the Lorentzian; we do not claim such a bound here.
The scaling $N\sim L^{3/2}/\sqrt{\varepsilon_{\log}}$ is supported by
numerical evidence (Table~\ref{tab:Ncompare}) but should be treated as an
engineering design rule.
\end{remark}

\subsection*{S0.3~~Derivation of the conservative screening bound}

We now derive the conservative screening bound
(Eqs.~S0.7--S0.8 below), which is stated inline in Sec.~II of the main text.

\begin{proposition}[Conservative log-error bound]\label{prop:conservative}
Under Assumptions~\ref{A1}--\ref{A5} (identical detuning, but
\emph{not} restricted to the flank-centered choice), fix $b>0$ and
choose the normalization $\tilde{y}(L)=1$.
Define $\phi(u)=-\ln(1+u^{2})$ and write
\[
\ln\tilde{y}(I)=N\bigl[\phi(a+bI)-\phi(a+bL)\bigr].
\]
The target in this normalization is $(I-L)$.
Denoting the residual $r(I)=\ln\tilde{y}(I)-(I-L)$, we have $r(L)=0$
and $r(0)=N[\phi(a)-\phi(a+bL)]+L$.

For any choice of $a$ such that the operating range
$\{a+bI : I\in[0,L]\}$ lies in the region where $\phi$ is concave
(i.e., $\phi''(u)\le 0$ throughout), the worst-case log-error satisfies
\begin{equation}
E_\infty
\;\le\;
\frac{N\|\phi''\|_\infty\, b^{2}L^{2}}{8}
\;+\;
\left|\frac{N\phi'(a+bL)\cdot b - 1}{2}\right|\cdot L,
\tag{S0.5}
\end{equation}
where $\|\phi''\|_\infty=\sup_{u\in[a,\,a+bL]}|\phi''(u)|$.
\end{proposition}
\begin{proof}[Derivation sketch]
Write $h(I)\equiv N\phi(a+bI)$. The slope is
$h'(I)=Nb\,\phi'(a+bI)$.
At $I=L$, we want the slope to match the target slope~$1$; define
the slope mismatch $\Delta s\equiv h'(L)-1=Nb\,\phi'(a+bL)-1$.
By the mean-value theorem on $[0,L]$:
\[
r(I)-r(L)=r(I)=\bigl[h(I)-h(L)\bigr]-(I-L)
=\int_{I}^{L}\bigl[1-h'(t)\bigr]\,dt.
\]
Write $1-h'(t)=(1-h'(L))+(h'(L)-h'(t))=-\Delta s + \int_{t}^{L}h''(s)\,ds$.
Since $h''(s)=Nb^{2}\phi''(a+bs)$, we bound
$|h''(s)|\le Nb^{2}\|\phi''\|_\infty$.
Integrating twice and applying the triangle inequality gives (S0.5).
\end{proof}

\begin{corollary}[Main-text conservative bound]\label{cor:conservative}
Under slope matching at $I=L$ (i.e., $Nb\,\phi'(a+bL)=1$, so
$\Delta s=0$), and using $\|\phi''\|_\infty\le 2$ (which holds
since $|\phi''(u)|=|2(u^{2}-1)/(1+u^{2})^{2}|\le 2$ for all
$u\in\mathbb{R}$), the bound simplifies to
\begin{equation}
E_\infty
\;\le\;
\frac{Nb^{2}L^{2}}{4}.
\tag{S0.6}
\end{equation}
Using $b=1/N$ (the slope-matching choice from $Nb=1$) gives
$E_\infty\le L^{2}/(4N)$.
If instead we retain a general $b$ but add the penalty from
imperfect slope matching (e.g., from the constraint $b\le b_{\max}$),
a combined conservative bound is
\begin{equation}
E_\infty
\;\le\;
\frac{L^{2}}{4N}+\frac{1}{2b^{2}N},
\tag{S0.7}
\end{equation}
which provides a conservative heuristic bound on the log-error.
Setting this $\le\ln(1+\varepsilon)$ and solving for $N$ yields the
conservative screening depth:
\begin{equation}
N_{\mathrm{safe}}
\;\ge\;
\left\lceil
\frac{L^{2}/4+1/(2b^{2})}{\ln(1+\varepsilon)}
\right\rceil.
\tag{S0.8}
\end{equation}
\end{corollary}

\begin{remark}[Status of the conservative bound]
Equation~(S0.7) is a \textbf{conservative
heuristic design rule}. It is conservative because: (i) we use a
global upper bound $\|\phi''\|_\infty\le 2$ instead of the actual
curvature, (ii) we do not exploit the minimax-optimal $C$ shift.
It is heuristic (not a formal guarantee) because: (i) the
derivation assumes the operating range lies in the concavity region
of $\phi$, which may not hold for all detuning choices; (ii) the
second term $1/(2b^{2}N)$ arises from a simplified penalty model
for flank-curvature mismatch that has not been proved to be a
rigorous upper bound in all parameter regimes.
$N_{\mathrm{safe}}$ from Eq.~(S0.8) is therefore a
\textbf{screening estimate}, suitable for preliminary design-space
exploration but not a certified minimax guarantee.
\end{remark}

\subsection*{S0.4~~Validity scope and failure cases}

The derivations above hold under the stated assumptions. We now
identify the regimes where each assumption may break down.

\begin{enumerate}
\item[\textbf{(V1)}]
\textbf{Lorentzian model (A\ref{A1}).}
The single-ring Lorentzian form $T=[1+(a+bI)^{2}]^{-1}$ is a
near-resonance approximation valid when the probe frequency is within
a few linewidths of the resonance. Far from resonance, higher-order
dispersion, Fano interference, or multi-mode effects introduce
deviations. \emph{Failure case:} operation with very large detuning
($|a+bI|\gg 1$ across the full interval), where the Lorentzian tails
may not be accurate for high-$Q$ rings.

\item[\textbf{(V2)}]
\textbf{Multiplicative cascade (A\ref{A2}).}
Requires that inter-ring reflections and back-scattering are negligible
(forward-propagating coupling only, i.e.\ negligible back-reflection at each inter-ring junction). \emph{Failure case:} very high ring count
$N$ with non-negligible back-reflection per stage, which can produce
Fabry--P\'erot-like ripples in the cascade transfer function.

\item[\textbf{(V3)}]
\textbf{Identical-detuning family (A\ref{A3}).}
The Taylor expansion and conservative bound both assume
$a_1=\cdots=a_N$. In practice, fabrication variations introduce
per-ring detuning spread $\sigma_a$. The Monte Carlo analysis in
Sec.~S8 quantifies robustness, but the \emph{analytical} bounds
(S0.2)--(S0.8) are strictly valid only for identical detuning.

\item[\textbf{(V4)}]
\textbf{Linear control-to-resonance mapping (A\ref{A4}).}
The linearized model $\omega_0(I)=\omega_0^{(0)}+\eta I$ introduces
systematic error at large control amplitudes. For carrier-injection
(free-carrier plasma effect) or thermal tuning over wide ranges,
second-order nonlinearity in the control-to-detuning mapping can
exceed $1\%$. \emph{Failure case:} large $L$ requiring a control
swing exceeding the linearity range of the tuning mechanism.

\item[\textbf{(V5)}]
\textbf{Finite interval (A\ref{A5}).}
All bounds scale with $L$ (typically as $L^{2}$ or $L^{3/2}$).
As $L\to\infty$, $N$ grows without bound and insertion loss
accumulates ($\sim N\cdot IL_{\mathrm{stage}}$), eventually degrading
the probe SNR below the useful regime. There is no finite $N$ that
works for all $L$ simultaneously. \emph{Practical regime:}
$L\lesssim 10$--$12$ (consistent with $L_{\mathrm{eff}}$ at p90--p95
from Sec.~S3) is the primary target; $L\gtrsim 16$ requires
$N\gtrsim 30$ even for moderate tolerance, pushing loss budgets.

\item[\textbf{(V6)}]
\textbf{Flank-centered initialization (A\ref{A6}).}
The Taylor-based scaling (Theorem~\ref{thm:scaling}) relies on the
cancellation $\phi''(-1)=0$ at the half-maximum point. If the
operating point deviates (e.g., due to fabrication offset pushing
$a+bI_0$ away from $-1$), a nonzero quadratic residual appears and
the effective scaling worsens to $E_\infty\sim L^{2}/N$ rather
than $L^{3}/N^{2}$. \emph{Mitigation:} heater/bias trimming to
restore the flank condition.
\end{enumerate}

\subsection*{S0.5~~Mapping to main-text equations}

\noindent
For reference, the results derived here correspond to the following
main-text equations:
\begin{itemize}
\item \textbf{Slope bound} (Lemma~\ref{lem:slope}):
rigorous; corresponds to main-text Eqs.~(\ref{eq:slopebound})--(\ref{eq:Nb}).
This is a \emph{guaranteed} necessary condition.

\item \textbf{Engineering $N$-estimate} (Theorem~\ref{thm:scaling}):
heuristic scaling with empirical prefactor $\kappa\approx 0.07$;
corresponds to main-text Eq.~(\ref{eq:Nestimate}).
This is a \emph{heuristic design rule} calibrated against numerical fits.

\item \textbf{Conservative bound} (Corollary~\ref{cor:conservative}):
conservative but not rigorously certified as a minimax upper bound;
derived as Eq.~(S0.7) in this supplement, stated inline in Sec.~II.
This is a \emph{conservative heuristic screening condition}.

\item \textbf{$N_{\mathrm{safe}}$}
(Corollary~\ref{cor:conservative}, Eq.~S0.8):
the safe screening depth derived from the conservative bound;
derived as Eq.~(S0.8) in this supplement, stated inline in Sec.~II.
This is a \emph{conservative backstop estimate for preliminary design}.
\end{itemize}

\noindent\textbf{Summary of guarantee status:}\\[2pt]
\begin{tabular}{lll}
\hline
Result & Status & Main-text Eq.\\
\hline
Slope bound $N|b|\ge 1$ & Rigorous (proved) & (\ref{eq:Nb})\\
Scaling $N\sim\kappa L^{3/2}/\sqrt{\varepsilon_{\log}}$
  & Heuristic (Taylor truncation $+$ empirical $\kappa$) & (\ref{eq:Nestimate})\\
Bound $E_\infty\le L^{2}/(4N)+1/(2b^{2}N)$
  & Conservative heuristic & (S0.7)\\
$N_{\mathrm{safe}}$ screening depth
  & Conservative backstop & (S0.8)\\
\hline
\end{tabular}


\section*{S1. Depth-scaling derivation and conservative screening bound}

This section provides the detailed derivations underlying the depth-scaling
relations and conservative screening bounds summarized in the main text
(Sec.~II). These results complement the rigorous treatment in Sec.~S0.

\subsection*{S1.1~~Local expansion and exponential-like behavior}

To provide immediate local intuition (without changing the global minimax objective), let $\delta=I-I_0$ around the flank-centered point $I_0=L/2$ and impose $a+bI_0=-1$.
With the local normalization $C=2^{N}$ (so that $\tilde y(I_0)=1$), a third-order expansion of $\tilde y(I)=C[1+(a+bI)^2]^{-N}$ gives
\begin{equation}
\begin{aligned}
\tilde y(I)&\approx 1+Nb\,\delta+\frac{N^2}{2}b^2\delta^2+\frac{N(N^2-1)}{6}b^3\delta^3+\mathcal{O}(\delta^4),
\end{aligned}
\tag{S1.1}
\end{equation}
so with $b\sim 1/N$, the linear and quadratic coefficients align with those of $e^{\delta}=1+\delta+\delta^2/2+\delta^3/6+\cdots$, explaining why the initialization is already close before refinement.

\subsection*{S1.2~~Log-domain analysis and scaling derivation}

For depth scaling, the logarithmic domain is more transparent. Under the same flank centering ($a+bI_0=-1$), expand around $I_0=L/2$ with $\delta=I-I_0$ to obtain
\begin{equation}
\ln \tilde y(I)=\mathrm{const}+Nb\,\delta+\frac{Nb^3}{6}\,\delta^3+\mathcal{O}(\delta^4).
\tag{S1.2}
\end{equation}
At $a+bI_0=-1$, the quadratic term cancels identically in the log expansion; imposing slope matching ($Nb=1$) gives
\begin{equation}
\ln \tilde y(I)=\mathrm{const}+\delta+\frac{\delta^3}{6N^2}+\mathcal{O}(\delta^4).
\tag{S1.3}
\end{equation}
Hence the leading log-domain residual scales as $r(\delta)\sim \delta^3/N^2$. Over $I\in[0,L]$ with $|\delta|\le L/2$, this implies $E_\infty\sim L^3/N^2$. Requiring $E_\infty\le \varepsilon_{\log}$ leads to
\begin{equation}
N\propto \frac{L^{3/2}}{\sqrt{\varepsilon_{\log}}},
\tag{S1.4}
\end{equation}
which explains the scaling used in the main-text engineering estimate (Eq.~(\ref{eq:Nestimate})). This derivation is heuristic (not a formal guarantee), and the prefactor remains platform- and fitting-criterion dependent.

\subsection*{S1.3~~Conservative upper bound and screening depth}

For fixed $b$ and the identical-detuning family ($a_1=\cdots=a_N\equiv a$), one can write a conservative heuristic condition for achieving a prescribed log-tolerance.
A simple normalization is to enforce $\tilde y(L)=1$ (matching the target $f(L)=1$).
For a particular constructive choice of $a$ that keeps $(a+bI)$ large and negative across $[0,L]$, one can bound the worst-case log-error as
\begin{equation}
E_\infty \ \le\ \frac{L^2}{4N}+\frac{1}{2b^2 N}.
\tag{S1.5}
\end{equation}
(This is a conservative rule of thumb; obtaining a formal guarantee would require a separate proof.)
As a screening estimate (not a formal guarantee), one may use
\begin{equation}
N \ \ge\ \left\lceil \frac{L^2/4 + 1/(2b^2)}{\ln(1+\varepsilon)}\right\rceil.
\tag{S1.6}
\end{equation}
While this bound is typically pessimistic, it provides a conservative backstop-style estimate for preliminary design screening.
The rigorous derivation of these bounds, including the concavity conditions and slope-matching assumptions, is given in Sec.~S0.3.

\clearpage
\section*{S2. Worked example and empirical logit-range calibration}

This section provides the detailed worked example for the input-to-output
mapping and the empirical logit-range calibration tables referenced in the
main text (Sec.~III).

\subsection*{S2.1~~Worked input-to-output mapping example}

As a worked example, consider
\begin{equation}
x=[-3.2,\ 1.2,\ 4.8,\ -0.9].
\tag{S2.1}
\end{equation}
Compute $m=\max x_n=4.8$.
Then $u=x-m=[-8.0,-3.6,0,-5.7]$ and $L=-\min u_n=8.0$.
The mapped control-signal levels are
\begin{equation}
I=u+L=[0,\ 4.4,\ 8.0,\ 2.3],
\tag{S2.2}
\end{equation}
and the required normalized exponentials are $e^{x_n-m}=e^{u_n}=e^{I_n-L}$.
Using the fitted model directly,
\begin{equation*}
T_k(I_n)=\frac{1}{1+(a_k+bI_n)^2},
\qquad
y(I_n)=\prod_{k=1}^{N}T_k(I_n).
\end{equation*}
Under the identical-detuning fit ($a_1=\cdots=a_N\equiv a$), this becomes
\begin{equation*}
\tilde y(I_n)=C\,y(I_n)=C\left[\frac{1}{1+(a+bI_n)^2}\right]^N.
\end{equation*}
For the re-fitted parameters used in this example,
\begin{equation}
\begin{aligned}
a&=-1.4588,\quad b=0.10202,\\
N&=10,\quad C=3.0896\times 10^{1}.
\end{aligned}
\tag{S2.3}
\end{equation}
which gives
\begin{equation}
\begin{aligned}
\tilde y(I_n)&=C\left[\frac{1}{1+(a+bI_n)^2}\right]^N,\\
&\approx [3.44\times 10^{-4},\ 2.73\times 10^{-2},\\
&\qquad\ \ 9.74\times 10^{-1},\ 3.26\times 10^{-3}].
\end{aligned}
\tag{S2.4}
\end{equation}

For reference, the corresponding target terms are
\begin{equation}
I_n-L=[-8.0,\ -3.6,\ 0,\ -5.7],
\tag{S2.5}
\end{equation}
and
\begin{equation}
\begin{aligned}
\big[e^{I_n-L}\big]&\approx \big[3.35\times 10^{-4},\ 2.73\times 10^{-2},\\
&\qquad\ \ 1.00,\ 3.35\times 10^{-3}\big].
\end{aligned}
\tag{S2.6}
\end{equation}

\begin{table}[h!]
\centering
\caption{Example ($N=10$): approximating $e^{x_n-m}=e^{I_n-L}$ using $\tilde y(I)=C[1+(a+bI)^2]^{-N}$ with parameters re-fitted on $I\in[0,8.0]$ using the same minimax pipeline.}
\label{tab:S-example}
\begin{ruledtabular}
\begin{tabular}{rrrrr}
$x_n$ & $I_n$ & target $e^{x_n-m}$ & approx $\tilde y(I_n)$ & rel.\ err. \\
\midrule
$-3.2$ & $0.0$ & $3.3546\times10^{-4}$ & $3.4443\times10^{-4}$ & $2.673\%$ \\
$ 1.2$ & $4.4$ & $2.7324\times10^{-2}$ & $2.7325\times10^{-2}$ & $0.004\%$ \\
$ 4.8$ & $8.0$ & $1.0000$ & $0.9739$ & $2.608\%$ \\
$-0.9$ & $2.3$ & $3.3460\times10^{-3}$ & $3.2585\times10^{-3}$ & $2.614\%$ \\
\end{tabular}
\end{ruledtabular}
\end{table}

\subsection*{S2.2~~Effective-range percentiles and clipping calibration}

We first estimate the logit range observed in data and then choose clipping accordingly.
From two autoregressive Transformers (distilgpt2 and gpt2) and two public corpora (Tiny Shakespeare and \textit{Pride and Prejudice})~\cite{vaswani2017attention,radford2019gpt2,hf_distilgpt2_card,karpathy_tinyshakespeare,gutenberg_pride_prejudice} at context length 128, the effective range
\begin{equation}
L_{\mathrm{eff},\alpha}=\max(\log p_{\mathrm{kept}})-\min(\log p_{\mathrm{kept}}),\qquad \alpha=0.999,
\tag{S2.7}
\end{equation}
fell in a relatively narrow band, summarized in Table~\ref{tab:S-effective_range}.

\begin{table}[h!]
\centering
\small
\caption{Effective-range percentiles ($L_{\mathrm{eff},0.999}$) at context length 128.}
\label{tab:S-effective_range}
\begin{tabular}{lcc}
\toprule
Percentile & All runs (4 runs) & GPT-2 \\
\midrule
p50 & $6.92$--$7.23$ & $7.09$--$7.23$ \\
p90 & $8.60$--$8.75$ & $8.73$--$8.75$ \\
p95 & $8.97$--$9.12$ & $9.06$--$9.12$ \\
p99 & $9.50$--$9.69$ & $9.58$--$9.69$ \\
\bottomrule
\end{tabular}
\end{table}

We then test clipping on the same rows with
\begin{equation}
\begin{aligned}
E_{\mathrm{cum}}(t) &= \tfrac{1}{2}\|\mathrm{softmax}(u^{(t)})-\mathrm{softmax}(u)\|_1,\\
u^{(t)} &= \max(u,t),\quad u=s-\max(s).
\end{aligned}
\tag{S2.8}
\end{equation}
and require p99$\{E_{\mathrm{cum}}\}\le10^{-3}$ ($0.1\%$ budget).
This criterion is satisfied at $t=-12$ (p99 $\approx4.27\times10^{-4}$) and violated at $t=-11$ (p99 $\approx1.24\times10^{-3}$), so we set $t^*=-12$ ($N_{\mathrm{clip}}=12$).

In practice, we (i) estimate an effective $L$ from data, (ii) verify that fixed clipping keeps softmax error small, and (iii) choose representative design points (e.g., $L\approx8$ or $L\approx12$) while treating the clipped tail as negligible.
Full protocol details, clipping-sweep tables/plots, and per-run statistics are provided in Sec.~S3.

\subsection*{S2.3~~Illustrative synthetic range map}

As a design-space reference, we consider synthetic logit-range regimes using $L=\max(x)-\min(x)$ after $QK^\top/\sqrt{d_k}$ scaling.
These regimes are illustrative rather than corpus-level percentiles; using the same fitting pipeline, Table~\ref{tab:S-LLMrange} summarizes achievable approximation error versus depth.

\begin{table}[h!]
\centering
\caption{Synthetic softmax logit-range regimes ($L=\max(x)-\min(x)$) and fitted worst-case relative error (design-space illustration; not intended as corpus-level statistics).}
\label{tab:S-LLMrange}
\begin{ruledtabular}
\begin{tabular}{rcccc}
$L$ regime & $N=5$ & $N=10$ & $N=20$ & $N=30$ \\
\midrule
$L=8$ & $10.9\%$ & $2.68\%$ & $0.67\%$ & $0.30\%$ \\
$L=12$ & $40.0\%$ & $9.25\%$ & $2.27\%$ & $1.01\%$ \\
$L=16$ & $113\%$  & $23.0\%$ & $5.44\%$ & $2.41\%$ \\
\end{tabular}
\end{ruledtabular}
\end{table}

Table~\ref{tab:S-LLMrange} suggests a simple rule of thumb: the required depth depends mainly on the target $L$ regime. Near $L\approx 8$, moderate depth reaches a few-percent error, whereas $L\gtrsim 12$ typically requires deeper cascades to approach $<1\%$ error.

We include Table~\ref{tab:S-LLMrange} as a synthetic design map rather than an empirical benchmark.

\clearpage
\section*{S3. Empirical logit-range extraction from real Transformer runs}
We extracted empirical attention-logit ranges from real model runs to complement the synthetic $L$-regime map in the main text.
We used two open-source autoregressive Transformers (distilgpt2 and gpt2) and two public corpora (Tiny Shakespeare and \textit{Pride and Prejudice}), with context length 128 and causal masking.
For each valid attention row, if $p=\mathrm{softmax}(s)$ then the raw range is
\begin{equation}
L_{\mathrm{raw}}=\max(s)-\min(s)=\max(\log p)-\min(\log p),
\end{equation}
where max/min are taken over valid causal keys only.
Because very small tail probabilities can dominate $\min(\log p)$, we additionally report an effective range:
\begin{equation}
L_{\mathrm{eff},\alpha}=\max(\log p_{\mathrm{kept}})-\min(\log p_{\mathrm{kept}}),
\end{equation}
where keys are sorted by attention weight and retained until cumulative mass reaches $\alpha=0.999$.

To stay within a 16~GB RAM budget, we processed one model at a time, batch size 1, fixed windowing (stride 128), and streaming histogram quantiles.
Observed process RSS stayed below 1.24~GB in these runs.

\begin{table}[h!]
\centering
\small
\caption{Empirical global logit-range percentiles from real model--dataset runs (context length 128): raw vs effective ($\alpha=0.999$).}
\label{tab:s3_empirical_L}
\begin{tabular}{llcccccc}
\toprule
Model & Dataset & raw p95 & raw p99 & $L_{\mathrm{eff}}$ p50 & $L_{\mathrm{eff}}$ p90 & $L_{\mathrm{eff}}$ p95 & $L_{\mathrm{eff}}$ p99 \\
\midrule
distilgpt2 & tiny\_shakespeare & 22.82 & 69.00 & 7.10 & 8.60 & 8.97 & 9.50 \\
distilgpt2 & pride\_prejudice & 21.76 & 68.60 & 6.92 & 8.60 & 9.03 & 9.57 \\
gpt2 & tiny\_shakespeare & 25.48 & 43.34 & 7.23 & 8.73 & 9.06 & 9.58 \\
gpt2 & pride\_prejudice & 24.13 & 40.92 & 7.09 & 8.75 & 9.12 & 9.69 \\
\bottomrule
\end{tabular}
\end{table}

For quick linkage to the main manuscript: the effective-range summary quoted in the main text corresponds to this table (all runs: p50 $=6.92$--$7.23$, p90 $=8.60$--$8.75$, p95 $=8.97$--$9.12$, p99 $=9.50$--$9.69$), and the GPT-2 subset is p50 $=7.09$--$7.23$, p90 $=8.73$--$8.75$, p95 $=9.06$--$9.12$, p99 $=9.58$--$9.69$.

\begin{figure}[h!]
\centering
\includegraphics[width=0.88\linewidth]{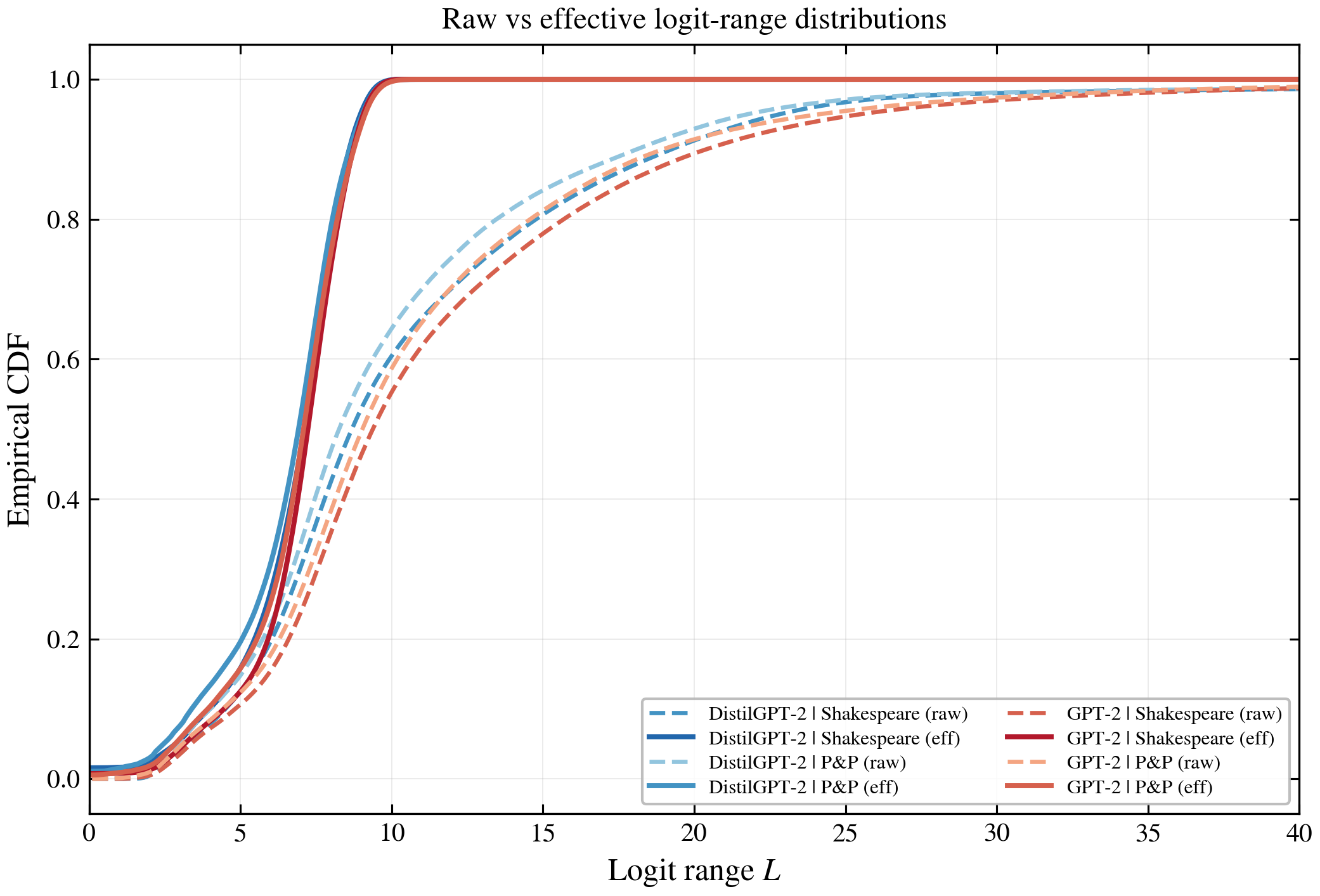}
\caption{Global CDFs of raw $L_{\mathrm{raw}}$ (dashed) and effective $L_{\mathrm{eff},0.999}$ (solid) for the four model--dataset runs.}
\label{fig:s3_cdf}
\end{figure}

\noindent\textbf{Clipping-validity sweep (additional justification).}
To test whether practical clipping magnitudes can be used without materially changing softmax outputs, we evaluated a thresholded-logit approximation.
For each row, define $u=s-\max(s)$ and, for threshold $t\le 0$,
\begin{equation}
u^{(t)}=\max(u,t),\qquad p^{(t)}=\mathrm{softmax}(u^{(t)}).
\end{equation}
We report the cumulative softmax error
\begin{equation}
E_{\mathrm{cum}}(t)=\frac{1}{2}\left\|p^{(t)}-p\right\|_1,
\end{equation}
then sweep $t\in\{-14,-13,\ldots,-6\}$ and compute p50/p90/p95/p99 of $E_{\mathrm{cum}}$ over all extracted rows.

\begin{table}[h!]
\centering
\small
\caption{Global clipping-validity sweep: percentile statistics of $E_{\mathrm{cum}}(t)$ versus clipping threshold $t$.}
\label{tab:s3_clip}
\begin{tabular}{ccccc}
\toprule
$t$ & p50 & p90 & p95 & p99 \\
\midrule
$-14$ & $2.53\times10^{-5}$ & $4.55\times10^{-5}$ & $4.80\times10^{-5}$ & $5.18\times10^{-5}$ \\
$-13$ & $2.69\times10^{-5}$ & $4.85\times10^{-5}$ & $7.38\times10^{-5}$ & $1.48\times10^{-4}$ \\
$-12$ & $2.99\times10^{-5}$ & $1.21\times10^{-4}$ & $2.13\times10^{-4}$ & $4.27\times10^{-4}$ \\
$-11$ & $3.31\times10^{-5}$ & $3.95\times10^{-4}$ & $6.55\times10^{-4}$ & $1.24\times10^{-3}$ \\
$-10$ & $3.72\times10^{-5}$ & $1.28\times10^{-3}$ & $2.01\times10^{-3}$ & $3.58\times10^{-3}$ \\
$-9$  & $4.41\times10^{-5}$ & $4.04\times10^{-3}$ & $6.11\times10^{-3}$ & $1.03\times10^{-2}$ \\
$-8$  & $2.25\times10^{-4}$ & $1.26\times10^{-2}$ & $1.83\times10^{-2}$ & $2.91\times10^{-2}$ \\
$-7$  & $2.76\times10^{-3}$ & $3.85\times10^{-2}$ & $5.30\times10^{-2}$ & $7.89\times10^{-2}$ \\
$-6$  & $1.88\times10^{-2}$ & $1.11\times10^{-1}$ & $1.43\times10^{-1}$ & $1.95\times10^{-1}$ \\
\bottomrule
\end{tabular}
\end{table}

\begin{figure}[h!]
\centering
\includegraphics[width=0.86\linewidth]{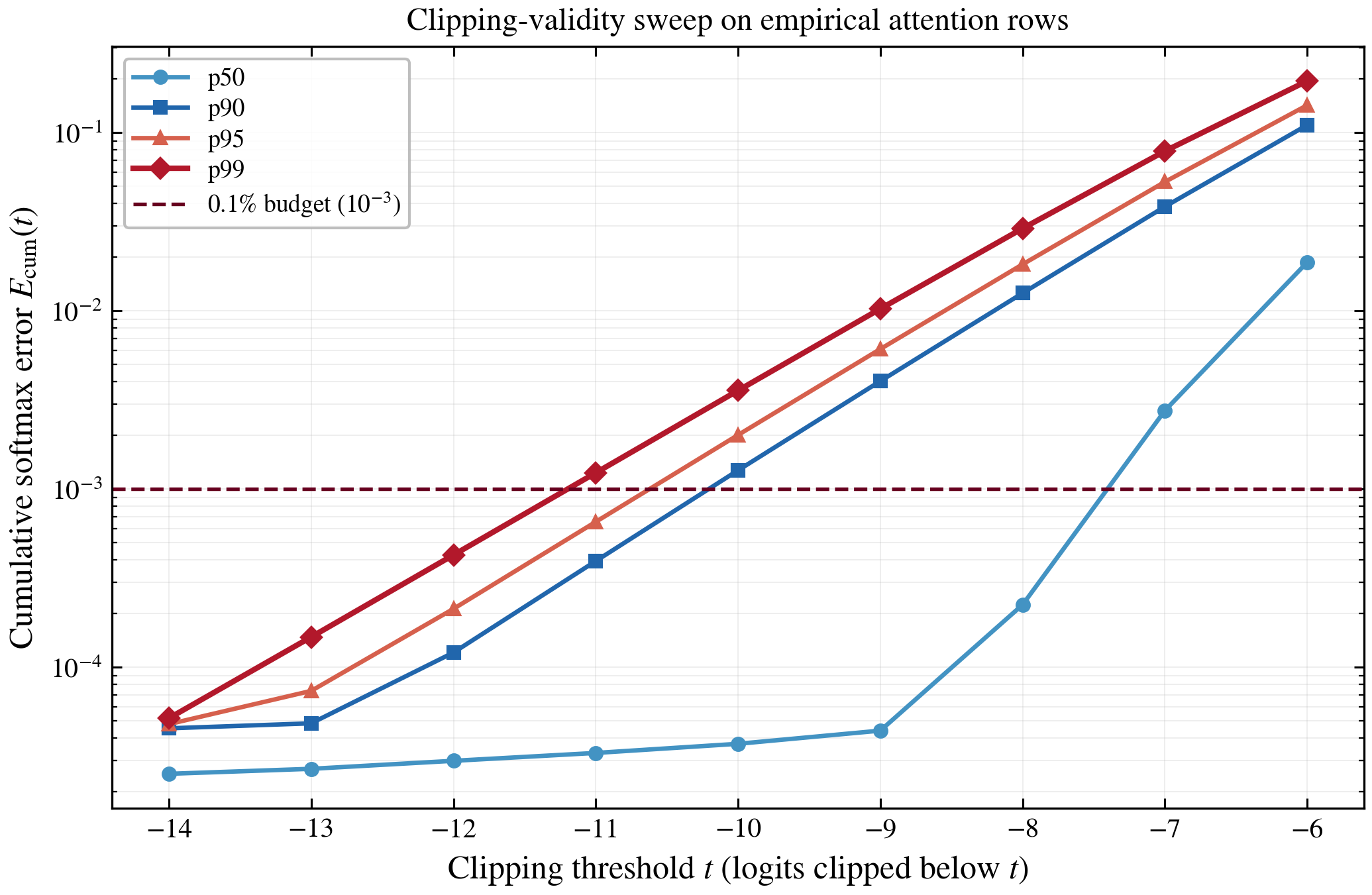}
\caption{Percentile curves of cumulative softmax error $E_{\mathrm{cum}}(t)$ versus clipping threshold $t$. The dashed line marks the $0.1\%$ budget ($10^{-3}$).}
\label{fig:s3_clip}
\end{figure}

Under a conservative budget criterion p99$\{E_{\mathrm{cum}}\}\le 10^{-3}$, the least negative admissible threshold in this sweep is $t^*=-12$ (p99 $\approx 4.27\times10^{-4}$).
Equivalently, the operational clipping magnitude is $N_{\mathrm{clip}}\equiv -t^*=12$.
Notably, this is closely aligned with the empirical effective-range scale (Table~\ref{tab:s3_empirical_L}: p99 of $L_{\mathrm{eff},0.999}$ up to $\approx9.69$), indicating that clipping-constrained implementation and effective-range statistics operate in the same order-of-magnitude range budget.
This supports using a practical clipping magnitude comparable to the design range scale ($L\approx N_{\mathrm{clip}}$) while keeping aggregate softmax distortion below $0.1\%$.

\clearpage
\section*{S4. FDTD methodology details and X-cut $b_V$ derivation}

This section provides the detailed FDTD simulation methodology, the
step-by-step X-cut arc electrode voltage sensitivity derivation, and the full
cascade optimization table referenced in the main text (Sec.~IV--V).

\subsection*{S4.1~~z-refined 3-fix simulation strategy}

For thin-film \ce{LiNbO3} structures, special care is required in the
vertical ($z$) direction due to the high index contrast between
\ce{LiNbO3} ($n_o \approx 2.21$) and \ce{SiO2} ($n \approx 1.44$)
and the sub-micron film thickness.  We apply a ``z-refined 3-fix''
strategy:
\begin{enumerate}
  \item \textbf{Ordinary index correction}: the material model uses
    the corrected ordinary refractive index $n_o$ appropriate for
    the TE mode in X-cut geometry, rather than the extraordinary
    index $n_e$ that governs TM propagation;
  \item \textbf{$z$-span expansion}: the simulation $z$-span is
    extended beyond the minimal waveguide region to include sufficient
    substrate and superstrate so that evanescent field tails
    are captured without PML truncation artifacts;
  \item \textbf{Auto-mesh}: accuracy level~3; conformal variant~1 meshing
    is enabled, and no manual mesh override is applied.
    The resulting vertical grid spacing in the slab region is
    approximately \SI{55}{\nano\meter}, providing ${\sim}2$ cells
    across the \SI{100}{\nano\meter} slab.
\end{enumerate}
This refinement strategy is critical for obtaining converged results
in TFLN ring resonators, where the high-$Q$ spectral features are
sensitive to numerical dispersion in under-resolved thin
films~\cite{zhu2021aop}.
Table~\ref{tab:fdtd_params} lists the full simulation parameters.

\begin{table}[h!]
\centering
\caption{3D FDTD simulation parameters (Lumerical).}
\label{tab:fdtd_params}
\begin{ruledtabular}
\begin{tabular}{ll}
Parameter & Value \\
\hline
Solver              & Lumerical 3D FDTD \\
Mesh type           & Conformal variant~1 \\
Mesh accuracy       & 3 (auto-mesh) \\
$z$-mesh override   & None (auto-mesh) \\
Simulation time     & \SI{50}{\pico\second} \\
Auto shutoff        & $1 \times 10^{-6}$ \\
Wavelength range    & \SIrange{1530}{1570}{\nano\meter} \\
Grid size           & $532 \times 816 \times 44$ \\
Source              & Broadband mode source (TE$_0$) \\
\end{tabular}
\end{ruledtabular}
\end{table}

\subsection*{S4.2~~X-cut arc electrode $b_V$ step-by-step derivation}

For the X-cut circular ring with lateral S--G arc electrodes
(Table~\ref{tab:device_params}), the crystal Z-axis (c-axis) is oriented
at $45^\circ$ from the horizontal axis in the substrate plane.  At
azimuthal angle $\theta$ around the ring, the projection of the
lateral electric field onto the Z-axis is proportional to
$\cos(\theta - 45^\circ)$.  The $\cos(\theta - 45^\circ) = 0$
boundaries fall at $\theta = 135^\circ$ and $\theta = 315^\circ$,
naturally separating the bus-waveguide coupling regions from the
electrode regions.  Each ring carries a full semicircular arc electrode
on the side opposite to its coupling points.  By the substitution
$\varphi = \theta - 45^\circ$, the effective EO fill factor is
\begin{equation}
  f_\mathrm{EO}
  = \frac{1}{2\pi}\int_\mathrm{semicircle} |\cos(\theta - 45^\circ)|\,d\theta
  = \frac{1}{2\pi}\int_{-\pi/2}^{+\pi/2} \cos\varphi\,d\varphi
  = \frac{1}{2\pi}\bigl[\sin\varphi\bigr]_{-\pi/2}^{+\pi/2}
  = \frac{1}{\pi}
  \approx 0.318.
  \tag{S4.1}
\end{equation}
The $45^\circ$ rotation ensures that the electrode semicircle does not
overlap with the coupling points, while the fill factor integral is
identical to the standard $\cos\theta$ case by the change of variable.

The lateral S--G electrodes have gap $g_\mathrm{el} = \SI{5}{\micro\meter}$,
giving an effective electrode--waveguide distance
$d_\mathrm{eff} \approx g_\mathrm{el}/2 = \SI{2.5}{\micro\meter}$.
The lateral field geometry yields an EO overlap factor
$\Gamma_\mathrm{EO} = 0.7$, compared to $0.5$ for a vertical
electrode configuration.

The refractive index change per volt in the electrode-covered section is
\begin{equation}
  \frac{\Delta n_\mathrm{eff}}{V}
  = -\frac{1}{2}\,n_e^3\,r_{33}\,
    \frac{\Gamma_\mathrm{EO}}{d_\mathrm{eff}}
  = -\frac{1}{2} \times 2.138^3 \times 30.9 \times 10^{-12}
    \times \frac{0.7}{2.5 \times 10^{-6}}
  = -4.226 \times 10^{-5}\;\mathrm{V^{-1}}.
  \tag{S4.2}
\end{equation}
The corresponding resonance wavelength shift is
\begin{equation}
  \left.\frac{d\lambda_0}{dV}\right|_\mathrm{straight}
  = \frac{1550 \times 4.226 \times 10^{-5}}{2.30}
  = \SI{28.48}{\pico\meter\per\volt},
  \tag{S4.3}
\end{equation}
giving an intrinsic (straight-section) voltage sensitivity of
\begin{equation}
  b_V^\mathrm{straight}
  = \frac{2Q_L}{\lambda_0}\left|\frac{d\lambda_0}{dV}\right|_\mathrm{straight}
  = \frac{2 \times 15{,}500}{1550} \times 0.02848
  = \SI{0.570}{\per\volt}.
  \tag{S4.4}
\end{equation}
However, only the arc-electrode portion of the ring circumference
contributes to the round-trip phase shift.  The effective voltage
sensitivity is therefore
\begin{equation}
  b_V
  = b_V^\mathrm{straight} \times f_\mathrm{EO}
  = 0.570 \times \frac{1}{\pi}
  \approx \SI{0.182}{\per\volt}.
  \tag{S4.5}
\end{equation}
A $\SI{1}{\volt}$ applied voltage shifts the normalized detuning by
$\Delta a \approx 0.182$.
Despite the fill-factor penalty ($f_\mathrm{EO} = 1/\pi \approx 0.318$),
the X-cut arc design benefits from a smaller effective electrode distance
($\SI{2.5}{\micro\meter}$ vs.\ $\SI{4}{\micro\meter}$ for vertical
configurations) and a higher overlap factor ($0.7$ vs.\ $0.5$),
which partially compensate the reduced active length.

\subsection*{S4.3~~Full cascade optimization table}

Table~\ref{tab:S-cascade} presents the complete optimization results
for the standard dynamic range $L = 8$ (corresponding to
$e^8 \approx 2981$, i.e., \SI{34.7}{\decibel}), covering all cascade
depths from $N=5$ to $N=30$.

\begin{table}[h!]
\centering
\caption{Cascade optimization results for $L = 8$.  The bias voltage
$V_\mathrm{bias} = |a|/b_V$ sets the DC offset, and
$V_\mathrm{ctrl} = bL/b_V$ is the maximum control voltage at
$I = L$.  Voltages computed with
$b_V = \SI{0.182}{\per\volt}$ (FDTD-calibrated best resonance
$Q_L = 15{,}500$).}
\label{tab:S-cascade}
\begin{ruledtabular}
\begin{tabular}{rcccccc}
$N$ & $a$ & $b$ & $E_\infty$ & $\varepsilon_\mathrm{max}$ (\%)
    & $V_\mathrm{bias}$ (\si{\volt}) & $V_\mathrm{ctrl}$ (\si{\volt}) \\
\hline
 5  & $-2.0789$ & 0.21658 & 0.1035 & 10.91 & 11.4 &  9.5 \\
 8  & $-1.5959$ & 0.12896 & 0.0412 &  4.20 &  8.8 &  5.7 \\
10  & $-1.4588$ & 0.10202 & 0.0265 &  2.68 &  8.0 &  4.5 \\
12  & $-1.3731$ & 0.08450 & 0.0184 &  1.86 &  7.5 &  3.7 \\
15  & $-1.2914$ & 0.06726 & 0.0118 &  1.19 &  7.1 &  3.0 \\
17  & $-1.2543$ & 0.05923 & 0.0092 &  0.92 &  6.9 &  2.6 \\
20  & $-1.2136$ & 0.05025 & 0.0067 &  0.67 &  6.7 &  2.2 \\
25  & $-1.1685$ & 0.04013 & 0.0043 &  0.43 &  6.4 &  1.8 \\
30  & $-1.1392$ & 0.03341 & 0.0030 &  0.30 &  6.3 &  1.5 \\
\end{tabular}
\end{ruledtabular}
\end{table}

Key thresholds for the minimum number of rings at various error
targets are:
\begin{itemize}
  \item $\varepsilon < 10\%$: $N \geq 6$,
  \item $\varepsilon < 5\%$:  $N \geq 8$,
  \item $\varepsilon < 2\%$:  $N \geq 12$,
  \item $\varepsilon < 1\%$:  $N \geq 17$,
  \item $\varepsilon < 0.5\%$: $N \geq 24$.
\end{itemize}
These thresholds are \emph{independent} of the quality factor $Q$,
since the minimax approximation operates entirely in normalized
detuning space.  The $Q$ factor affects only the physical voltage
required to achieve the necessary detuning range, through $b_V$.

\subsection*{S4.4~~Lorentzian fit validation}

Figure~\ref{fig:single_ring_lorentzian} shows the Lorentzian fit
to the FDTD drop-port resonance near $\lambda = \SI{1566}{\nano\meter}$.
The analytical Lorentzian
$T_\mathrm{drop}(\Delta\lambda) = A/[1 + (2Q\Delta\lambda/\lambda_0)^2]$
with $Q_L = 15{,}500$ closely tracks the FDTD data, validating the
single-ring transfer function model used in the cascade analysis.

\begin{figure}[h!]
\centering
\includegraphics[width=0.7\textwidth]{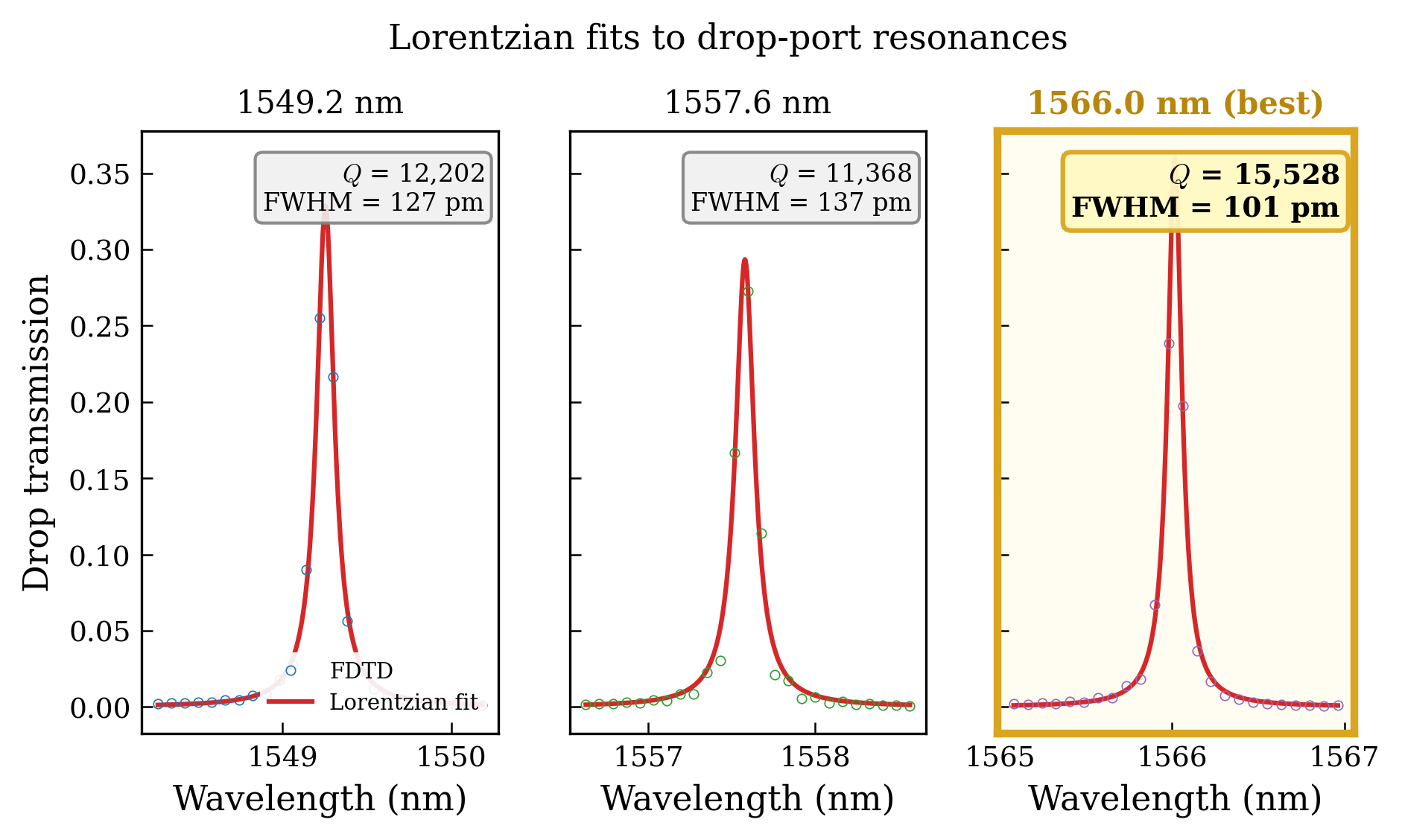}
\caption{Lorentzian fit to the FDTD drop-port resonance.  Markers:
FDTD data; solid line: Lorentzian fit.  The extracted quality factor
is $Q_L = 15{,}500$ with
$\mathrm{FWHM} = \SI{101}{\pico\meter}$.}
\label{fig:single_ring_lorentzian}
\end{figure}

\subsection*{S4.5~~Eigenmode (FDE) analysis of theoretical $Q_i$}

To quantify how far below the physical limit the FDTD-extracted
$Q_i = 38{,}800$ lies, we perform a two-dimensional
finite-difference eigenmode (FDE) analysis of the bent rib
waveguide cross-section using Lumerical MODE Solutions.

\paragraph{Setup.}
The FDE solver models the cross-section of the rib waveguide at
the design bend radius $R = \SI{20}{\micro\meter}$ and wavelength
$\lambda = \SI{1550}{\nano\meter}$, with perfectly matched layer
(PML) boundaries on all four edges.
The geometry is identical to the 3D FDTD model:
\SI{600}{\nano\meter} total \ce{LiNbO3}
($n_o = 2.211$, lossless dielectric),
\SI{100}{\nano\meter} slab, \SI{500}{\nano\meter} rib etch,
waveguide width $W = \SI{1.4}{\micro\meter}$, on a
\SI{2}{\micro\meter} \ce{SiO2} substrate ($n = 1.444$) with air
cladding.
The mesh is set to $300 \times 300$ cells over a
$\SI{6}{\micro\meter} \times \SI{3}{\micro\meter}$ cross-section,
yielding effective grid spacings
$\Delta x \approx \SI{20}{\nano\meter}$ and
$\Delta y \approx \SI{10}{\nano\meter}$---substantially finer than
the 3D FDTD auto-mesh (\SI{55}{\nano\meter} vertical).

\paragraph{Complex effective index.}
The FDE solver returns a complex effective index
$n_\mathrm{eff} = n_r + i\,n_i$ for each guided mode, where the
imaginary part $n_i$ encodes propagation loss.
For the fundamental TE mode at $R = \SI{20}{\micro\meter}$:
\begin{align}
  n_\mathrm{eff} &= 1.9653 + i\,(4.73 \times 10^{-8}),
  \label{eq:fde_neff} \\
  \alpha_\mathrm{rad+leak}
    &= \frac{4\pi\,n_i}{\lambda}
    = \SI{0.383}{\per\meter}
    \;\bigl(\SI{0.017}{\decibel\per\centi\meter}\bigr).
  \label{eq:fde_alpha}
\end{align}
Since the material is set as lossless, this $\alpha$ captures only
bending radiation loss and substrate leakage through the
\SI{100}{\nano\meter} slab.  The corresponding quality factor is
\begin{equation}
  Q_\mathrm{rad+leak}
  = \frac{2\pi\,n_g}{\alpha_\mathrm{rad+leak}\,\lambda}
  = 2.43 \times 10^7,
  \label{eq:fde_Q_radleak}
\end{equation}
where $n_g = 2.354$ is the group index from the FDE solver (consistent
with the FDTD FSR-derived $n_g = 2.30$; the small difference arises
from the straight-section approximation inherent to 2D FDE).

\paragraph{Decomposition into bending and leakage.}
A separate FDE run with $R = \SI{1}{\milli\meter}$ (effectively
straight) yields $Q_\mathrm{leak} = 2.93 \times 10^7$, isolating
the substrate leakage contribution.  The pure bending radiation
quality factor follows from
\begin{equation}
  \frac{1}{Q_\mathrm{bend}}
  = \frac{1}{Q_\mathrm{rad+leak}}
  - \frac{1}{Q_\mathrm{leak}},
  \quad
  Q_\mathrm{bend} = 1.43 \times 10^8.
  \label{eq:fde_Q_bend}
\end{equation}
This confirms that bending radiation loss at
$R = \SI{20}{\micro\meter}$ is negligible; substrate leakage
through the thin slab is the dominant geometric loss channel.

\paragraph{Material absorption.}
The FDE mode profile yields a confinement factor
$\Gamma = 0.887$ (fraction of the optical intensity within the
\ce{LiNbO3} core and slab regions).  The material-absorption-limited
quality factor is
\begin{equation}
  Q_\mathrm{abs}
  = \frac{2\pi\,n_g}{\Gamma\,\alpha_\mathrm{mat}\,\lambda},
  \label{eq:fde_Q_abs}
\end{equation}
where $\alpha_\mathrm{mat}$ is the bulk material power-attenuation
coefficient of \ce{LiNbO3} at \SI{1550}{\nano\meter}.
Table~\ref{tab:theoretical_Qi} evaluates Eq.~\eqref{eq:fde_Q_abs}
for representative TFLN absorption values from the
literature~\cite{zhu2021aop, hu2024integratedeo}.

\begin{table}[h!]
\centering
\caption{Theoretical intrinsic quality factor $Q_i$ of the
$R = \SI{20}{\micro\meter}$ TFLN ring, decomposed into radiation
($Q_\mathrm{bend}$), substrate leakage ($Q_\mathrm{leak}$), and
material absorption ($Q_\mathrm{abs}$).
Sidewall scattering (fabrication-dependent) is excluded.
The total is
$1/Q_i = 1/Q_\mathrm{rad+leak} + 1/Q_\mathrm{abs}$
with $Q_\mathrm{rad+leak} = 2.43 \times 10^7$.}
\label{tab:theoretical_Qi}
\begin{ruledtabular}
\begin{tabular}{lccc}
Material condition
  & $\alpha_\mathrm{mat}$ (dB/cm)
  & $Q_\mathrm{abs}$
  & $Q_i$ (total) \\
\hline
Bulk \ce{LiNbO3} (pristine)
  & 0.002 & $2.3 \times 10^8$ & $2.2 \times 10^7$ \\
High-quality TFLN
  & 0.01  & $4.7 \times 10^7$ & $1.6 \times 10^7$ \\
Good TFLN
  & 0.03  & $1.6 \times 10^7$ & $9.5 \times 10^6$ \\
Typical TFLN
  & 0.1   & $4.7 \times 10^6$ & $3.9 \times 10^6$ \\
\end{tabular}
\end{ruledtabular}
\end{table}

For high-quality TFLN
($\alpha_\mathrm{mat} \lesssim
  \SI{0.01}{\decibel\per\centi\meter}$),
the theoretical $Q_i$ exceeds $10^7$---more than $400\times$ higher
than the FDTD-extracted value of $38{,}800$.
This confirms that the FDTD result is dominated by numerical mesh
artifacts (approximately two cells across the \SI{100}{\nano\meter}
slab), not by physical loss mechanisms.
Bending radiation loss at $R = \SI{20}{\micro\meter}$ is negligible
($Q_\mathrm{bend} = 1.43 \times 10^8$); the dominant geometric
loss channel in the ideal structure is substrate leakage through
the thin slab ($Q_\mathrm{leak} = 2.93 \times 10^7$).

\clearpage
\section*{S5. Fabricated high-$Q$ design projections}

Reproducing $Q_i > 10^5$ in three-dimensional FDTD is
computationally impractical: at accuracy level~3 the
\SI{100}{\nano\meter} slab requires $\Delta z \lesssim
\SI{20}{\nano\meter}$ to suppress staircase-induced scattering,
inflating wall times beyond 30~days per run.  The numerically
extracted $Q_i = 38{,}800$ therefore represents a \emph{simulation}
floor, not a physical one.  A two-dimensional MODE-solver bend
analysis confirms $Q_\mathrm{bend} > 4.5 \times 10^7$ for
$R = \SI{20}{\micro\meter}$, placing bending radiation loss far
below any realistic intrinsic loss.

Table~\ref{tab:highQ_literature} surveys recent high-$Q$ TFLN
microring demonstrations.  These studies show that
$Q_i \ge 9 \times 10^6$ has been demonstrated in X-cut TFLN
using multiple fabrication routes, including Ar$^+$ milling,
wet etching, and ICP-RIE/CMP-based processes.

\begin{table}[h]
\centering
\caption{Demonstrated intrinsic quality factors in TFLN micro-ring
resonators.  ``EO compatible'' indicates whether the fabrication
process preserves electrode patterning capability.}
\label{tab:highQ_literature}
\begin{ruledtabular}
\begin{tabular}{lcccc}
Ref. & $Q_i$ & $R$ ($\mu$m) & $w$ ($\mu$m) & Etch \\
\hline
Zhang~\cite{zhang2017highQ_optica}
  & $10^7$              & 80  & ${\sim}2$   & Ar$^+$ mill \\
Gao~\cite{gao2022ultrahighQ}
  & $10^8$              & 100 & ${\sim}3$   & CMP$^\ast$ \\
Zhuang~\cite{zhuang2023wetetch}
  & $9{\times}10^6$     & 100 & ${\sim}2$   & Wet etch \\
Song~\cite{zhu2024highQ_racetrack}
  & $2.9{\times}10^7$   & 200 & 4.5         & ICP-RIE+CMP \\
\end{tabular}
\end{ruledtabular}
\vspace{-2pt}
{\footnotesize All processes except $^\ast$ are EO-electrode
compatible. $^\ast$CMP-only (no dry etch); subsequent electrode
patterning may degrade $Q_i$.}
\end{table}

To project cascade performance into the fabricated regime,
we fix $Q_\mathrm{ext} = 25{,}800$ (the FDTD-extracted coupling
quality factor at gap $= \SI{100}{\nano\meter}$) and compute
$D_\mathrm{max} = [Q_i/(Q_i + Q_\mathrm{ext})]^2$ for three
representative intrinsic quality factors
(Table~\ref{tab:projected_cascade}).

\begin{table}[h]
\centering
\caption{Projected cascade transmission for fabricated $Q_i$ values
at fixed $Q_\mathrm{ext} = 25{,}800$.
$D_\mathrm{max}^N$ is the ideal on-resonance cascade transmission in dB.
The minimax approximation error $\varepsilon_\mathrm{max}$ depends only
on $N$ and $L$ (not on $Q_i$); at $N = 20$, $L = 8$:
$\varepsilon_\mathrm{max} = 0.67\%$ (Table~\ref{tab:Ncompare}).}
\label{tab:projected_cascade}
\begin{ruledtabular}
\begin{tabular}{llcccc}
Projection & $Q_i$ & $D_\mathrm{max}$ & $N{=}10$ & $N{=}20$
  & $N{=}30$ \\
\hline
FDTD baseline
  & $3.88{\times}10^4$ & 0.36 & $-44.3$ & $-88.5$ & $-132.8$ \\
Conservative
  & $5{\times}10^5$    & 0.90 & $-4.4$  & $-8.8$  & $-13.2$ \\
Moderate
  & $10^6$             & 0.95 & $-2.2$  & $-4.5$  & $-6.7$ \\
Optimistic
  & $5{\times}10^6$    & 0.99 & $-0.44$ & $-0.88$ & $-1.3$ \\
\end{tabular}
\end{ruledtabular}
\end{table}

Even in the conservative scenario ($Q_i = 5 \times 10^5$),
$D_\mathrm{max} = 0.90$ and the $N = 10$ cascade loss is only
$-4.4\;\mathrm{dB}$---an order-of-magnitude improvement over the
FDTD baseline.  The moderate projection ($Q_i = 10^6$) matches the
``fabricated high-$Q$'' column in
Table~\ref{tab:power_budget}.  Because
$Q_\mathrm{bend} \approx 4.5 \times 10^7 \gg Q_i$ for all
projections, bending loss is never the bottleneck; the dominant
loss mechanism is sidewall scattering, which is determined entirely
by fabrication quality.  The literature values in
Table~\ref{tab:highQ_literature} support the view that intrinsic quality factors in the projected range are physically
achievable in TFLN---albeit with wider waveguides
($w \ge \SI{2}{\micro\meter}$) and larger ring radii
($R \ge \SI{80}{\micro\meter}$) than the present design.
Transferring comparable sidewall quality to our geometry
($R = \SI{20}{\micro\meter}$, $w = \SI{1.4}{\micro\meter}$) is an
open fabrication challenge; the projections in
Table~\ref{tab:projected_cascade} should be read as design targets
contingent on achieving it.

\clearpage
\section*{S6. Insertion loss budget details}

For a cascade of $N$ rings, the total insertion loss is modeled as
\begin{equation}
  IL_{\mathrm{tot}} \approx N \cdot IL_{\mathrm{stage}}
    + IL_{\mathrm{coupling}},
  \tag{S6.1}
\end{equation}
where $IL_{\mathrm{stage}}$ is the per-ring insertion loss at off-resonance
operation and $IL_{\mathrm{coupling}}$ accounts for fiber-to-chip and
chip-to-fiber coupling losses.
Using typical loss numbers from the
literature~\cite{wang2024siliconroadmap,leijtens2018multimode,li2023inmemoryphotonic,vermeulen2010grating,tan2022microringloss},
we consider two scenarios:

\begin{itemize}
\item \textbf{Optimistic:}
  $IL_{\mathrm{stage}} = 0.08\;\mathrm{dB}$,
  $IL_{\mathrm{coupling}} = 1.5\;\mathrm{dB}$.
  Then $IL_{\mathrm{tot}} \approx 1.90\;\mathrm{dB}$ ($N = 5$),
  $2.30\;\mathrm{dB}$ ($N = 10$),
  $3.10\;\mathrm{dB}$ ($N = 20$),
  and $3.80\;\mathrm{dB}$ ($N = 30$).

\item \textbf{Conservative:}
  $IL_{\mathrm{stage}} = 0.25\;\mathrm{dB}$,
  $IL_{\mathrm{coupling}} = 3.0\;\mathrm{dB}$.
  Then $IL_{\mathrm{tot}} \approx 4.25\;\mathrm{dB}$ ($N = 5$),
  $5.50\;\mathrm{dB}$ ($N = 10$),
  $8.00\;\mathrm{dB}$ ($N = 20$),
  and $10.5\;\mathrm{dB}$ ($N = 30$).
\end{itemize}

In both scenarios, $N = 5$--$10$ is manageable for probe-power budgeting,
whereas $N = 20$ and $N = 30$ require tighter power budgeting and more
amplification margin.  Higher $IL_{\mathrm{tot}}$ raises the required
probe SNR and pushes operation closer to the detector noise floor,
reducing usable dynamic range.

\paragraph{Four-component loss breakdown.}
The total insertion loss of the cascade has four components:
\begin{enumerate}
  \item \textbf{On-resonance cascade transmission} $D_\mathrm{max}^N$
    (dominant; see Table~\ref{tab:power_budget});
  \item \textbf{Inter-ring coupling loss}
    $(N-1)\times(-10\log_{10}\eta_\mathrm{coupling})$, where
    $\eta_\mathrm{coupling}$ is the power transfer efficiency at each
    inter-ring bus section.  Two-ring FDTD yields
    $\eta_\mathrm{coupling} \approx 0.9$ for the present
    diagonal-bus geometry, corresponding to
    ${\sim}0.46\;\mathrm{dB}$ per inter-ring stage;
  \item \textbf{Off-resonance propagation loss}
    $N \times IL_\mathrm{stage}$, where
    $IL_\mathrm{stage} = 0.08$--$0.25\;\mathrm{dB}$ per
    ring~\cite{wang2024siliconroadmap,leijtens2018multimode,li2023inmemoryphotonic,tan2022microringloss};
  \item \textbf{Fiber-to-chip coupling loss}
    $IL_\mathrm{coupling} = 1.5$--$3.0\;\mathrm{dB}$~\cite{vermeulen2010grating}.
\end{enumerate}
Table~\ref{tab:power_budget} presents the \emph{ideal} on-resonance
budget ($D_\mathrm{max}^N$ only).  Including all four components for
the present diagonal-bus layout:
in the FDTD-characterized regime ($D_\mathrm{max} = 0.36$, $N = 5$),
the total loss is approximately $22.2 + 1.8 + 0.4 + 1.5 \approx
26\;\mathrm{dB}$;
in the fabricated high-$Q$ regime ($D_\mathrm{max} = 0.95$, $N = 30$),
the total loss is $6.7 + 13.3 + 2.4 + 1.5 \approx 24\;\mathrm{dB}$.
The inter-ring coupling loss dominates in the high-$Q$ regime,
underscoring that layout optimization (e.g., adiabatic tapers or
straight-bus coupling) is as important as achieving
$D_\mathrm{max} \geq 0.95$ through quality-factor improvement.
For an optimized layout with $\eta_\mathrm{coupling} \geq 0.98$
(${\le}0.09\;\mathrm{dB}$ per stage), the $N = 30$ total loss
would reduce to ${\sim}13\;\mathrm{dB}$.

\clearpage
\section*{S7. Energy efficiency detailed derivation}

This section provides the detailed energy-per-operation derivations for
both electrical analog exponential circuits and the photonic MRR cascade,
as summarized in the main text (Sec.~V).

\subsection*{S7.1~~Electrical analog exponential circuits}

Three main families of electrical circuits realize the exponential
function in the analog domain:

\paragraph{BJT translinear / Gilbert cell circuits.}
The collector current of a bipolar junction transistor is
$I_C = I_S \exp(V_{BE}/V_T)$, providing an intrinsic exponential
map~\cite{Gilbert1975,Mead1989}.
A Gilbert cell multiplier---the core building block of translinear
exponential circuits---dissipates $250$--$325\;\mu\mathrm{W}$ in
typical CMOS/BiCMOS implementations~\cite{GilbertCell2009}.
At a signal bandwidth of $B \approx 100\;\mathrm{MHz}$, the energy
per operation is
\begin{equation}
  E_{\mathrm{Gilbert}} = \frac{P}{B}
    = \frac{300\;\mu\mathrm{W}}{100\;\mathrm{MHz}}
    = 3\;\mathrm{pJ}.
  \tag{S7.1}
\end{equation}

\paragraph{CMOS subthreshold exponential circuits.}
A MOSFET in weak inversion exhibits
$I_D \propto \exp(V_{GS}/nV_T)$, enabling direct exponential
computation at ultra-low power~\cite{Mead1989}.
A reconfigurable softmax circuit in $180\;\mathrm{nm}$ CMOS
implements a 10-input softmax at
$V_{DD} = 500\;\mathrm{mV}$ with
$P = 3\;\mu\mathrm{W}$~\cite{Softmax180nm2021}.
Per-channel: $P_{\mathrm{exp}} \approx 0.43\;\mu\mathrm{W}$.
At $B \approx 1\;\mathrm{MHz}$ (limited by subthreshold $f_T$):
\begin{equation}
  E_{\mathrm{sub\text{-}V_T}}
    = \frac{0.43\;\mu\mathrm{W}}{1\;\mathrm{MHz}}
    = 0.43\;\mathrm{pJ}.
  \tag{S7.2}
\end{equation}
This is the most energy-efficient electrical approach, but at
severely limited bandwidth (${\sim}1\;\mathrm{MHz}$).

\paragraph{Digital CMOS (for reference).}
A digital exponential via Taylor series requires ${\sim}10$
multiply-add operations.
Using Horowitz's energy figures~\cite{Horowitz2014} for
$45\;\mathrm{nm}$ at $0.9\;\mathrm{V}$: 32-bit FP multiply costs
$3.7\;\mathrm{pJ}$, FP add costs $0.9\;\mathrm{pJ}$, giving
\begin{equation}
  E_{\mathrm{digital}} \approx 10 \times (3.7\;\mathrm{pJ}
    + 0.9\;\mathrm{pJ}) = 46\;\mathrm{pJ}.
  \tag{S7.3}
\end{equation}
At 8-bit precision (sufficient for inference): ${\sim}2.3\;\mathrm{pJ}$.

\subsection*{S7.2~~Photonic MRR cascade: single-channel energy derivation}

We evaluate the energy at $N = 30$ cascaded X-cut TFLN micro-ring
resonators with $R = 20\;\mu\mathrm{m}$ in the fabricated high-$Q$
regime ($Q_i = 10^6$, $Q_L \approx 25{,}200$; Supplementary Sec.~S5),
which achieves $\varepsilon_{\mathrm{max}} = 0.30\%$ with
$V_{\mathrm{ctrl}} = 0.91\;\mathrm{V}$ (fully CMOS-compatible).
The energy per exponential operation has three components:

\textit{(i) Electro-optic tuning energy.}
Each ring is tuned by charging the arc electrode capacitance to
$V_{\mathrm{ctrl}}$.  For the lateral S--G arc electrodes covering
one semicircle ($L_\mathrm{arc} = \pi R = \SI{62.8}{\micro\meter}$),
the electrode capacitance is estimated as
\begin{align}
  C_{\mathrm{el}} &\approx \SI{18}{\femto\farad},
  \tag{S7.4}
\end{align}
based on coplanar electrode modeling for TFLN lateral S--G
geometries with $g_\mathrm{el} = \SI{5}{\micro\meter}$
(comparable to values reported by
Bahadori~\textit{et al.}~\cite{bahadori2020oe} for similar
geometries).
The switching energy per ring at
$V_{\mathrm{ctrl}} = 0.91\;\mathrm{V}$
(using the projected $Q_L = 25{,}200$, which gives
$b_V = 0.295\;\mathrm{V}^{-1}$):
\begin{equation}
  E_{\mathrm{ring}} = \tfrac{1}{2}\,C_{\mathrm{el}}\,
    V_{\mathrm{ctrl}}^2
    = \tfrac{1}{2} \times 18\;\mathrm{fF}
      \times (0.91\;\mathrm{V})^2
    = 7.4\;\mathrm{fJ}.
  \tag{S7.5}
\end{equation}
For $N = 30$ rings:
$E_{\mathrm{EO}} = 30 \times 7.4 = 0.22\;\mathrm{pJ}$.

Note the important scaling: $E_{\mathrm{EO}} \propto 1/N$ since
$b \propto 1/N$ from minimax optimization, because
\begin{equation}
  E_{\mathrm{EO}} \propto N \times V_{\mathrm{ctrl}}^2
    \propto N \times (b/b_V)^2 \propto 1/N.
  \tag{S7.6}
\end{equation}
The bias voltage ($3.9\;\mathrm{V}$) is static and does not contribute
per-operation energy.

\textit{(ii) Laser source energy (amortized).}
Because every cascade channel uses the same fixed probe wavelength,
a single CW laser can be shared among $M$ parallel softmax channels
via a $1 \times M$ optical power splitter.
With wall-plug efficiency
$\eta_{\mathrm{WPE}} \approx 15\%$~\cite{WDMLaser2013},
the per-channel optical power is
$P_{\mathrm{opt}} = P_{\mathrm{in}}/M \approx 100\;\mu\mathrm{W}$
(for $P_{\mathrm{in}} = 1\;\mathrm{mW}$, $M = 10$), requiring
$P_{\mathrm{laser}} \approx 667\;\mu\mathrm{W}$ per channel.
At $f_{\mathrm{mod}} = 10\;\mathrm{GHz}$:
$E_{\mathrm{laser}} = 667\;\mu\mathrm{W}\,/\,10\;\mathrm{GHz}
= 67\;\mathrm{fJ}$.

\textit{(iii) Photodetector energy.}
Integrated SiGe photodetectors with TIA achieve sub-pJ
reception~\cite{Miller2017attojoule}:
$E_{\mathrm{PD}} \approx 0.5\;\mathrm{pJ}$.

The total single-channel energy is
\begin{equation}
  E_{\mathrm{photonic}}^{(1\mathrm{ch})}
    = E_{\mathrm{EO}} + E_{\mathrm{laser}} + E_{\mathrm{PD}}
    = 0.22 + 0.07 + 0.50
    = 0.79\;\mathrm{pJ}.
  \tag{S7.7}
\end{equation}

\subsection*{S7.3~~$Q$-factor scaling of energy efficiency}

Since $V_{\mathrm{ctrl}} \propto 1/Q$ and
$E_{\mathrm{EO}} \propto V_{\mathrm{ctrl}}^2$,
the EO energy scales as $1/Q^2$.
Table~\ref{tab:energy_vs_Q} shows the total energy for $N = 30$ at
various quality factors.

\begin{table}[h!]
\centering
\caption{Energy per exponential operation vs.\ quality factor
($N = 30$, $\varepsilon_{\mathrm{max}} = 0.30\%$,
X-cut arc electrode with $b_V$ scales linearly with $Q$;
$C_\mathrm{el} = 18\;\mathrm{fF}$).
$E_{\mathrm{laser}} + E_{\mathrm{PD}} = 0.57\;\mathrm{pJ}$
is the $Q$-independent floor.  The dagger ($\dagger$) marks the FDTD-calibrated quality factor;
the double dagger ($\ddagger$) marks the high-$Q$ design point
($Q_i = 10^6$).
Excludes thermal stabilization ($0.15$--$0.60\;\mathrm{pJ}$ for $N=30$).}
\label{tab:energy_vs_Q}
\begin{ruledtabular}
\begin{tabular}{rcccc}
$Q$ & $V_{\mathrm{ctrl}}$ (V) & $V_{\mathrm{bias}}$ (V)
    & $E_{\mathrm{EO}}$ (pJ) & $E_{\mathrm{total}}$ (pJ) \\
\hline
  5{,}000  & 4.57 & 19.5 & 5.64 & 6.21 \\
 10{,}000  & 2.28 &  9.7 & 1.40 & 1.97 \\
 12{,}500  & 1.83 &  7.8 & 0.90 & 1.47 \\
 \textbf{15{,}500}$^\dagger$ & \textbf{1.47} & \textbf{6.3} & \textbf{0.58} & \textbf{1.15} \\
 20{,}000  & 1.14 &  4.9 & 0.35 & 0.92 \\
 \textbf{25{,}200}$^\ddagger$ & \textbf{0.91} & \textbf{3.9} & \textbf{0.22} & \textbf{0.79} \\
 30{,}000  & 0.76 &  3.2 & 0.16 & 0.73 \\
 50{,}000  & 0.46 &  1.9 & 0.06 & 0.63 \\
\end{tabular}
\end{ruledtabular}
\end{table}

At $Q_L = 15{,}500$ (FDTD-calibrated), the EO contribution
($0.58\;\mathrm{pJ}$) is comparable to the optical floor,
placing the design in the efficient operating regime.
Beyond $Q \approx 30{,}000$, the EO contribution becomes negligible
and the total energy saturates near the floor; further $Q$ improvement
primarily benefits CMOS driver voltage compatibility rather than energy.

\paragraph{Additional energy contributions.}
The estimates above exclude two further contributions:
(i)~DAC energy for setting the per-ring control voltages, typically
$0.1$--$1\;\mathrm{pJ}$ per conversion at $10\;\mathrm{GHz}$
bandwidth; and
(ii)~thermal stabilization power for maintaining resonance alignment,
estimated at ${\sim}50$--$200\;\mu\mathrm{W}$ per ring for TFLN
(lower than silicon due to the small thermo-optic coefficient of
\ce{LiNbO3}, $dn/dT \approx 3.9 \times 10^{-6}\;\mathrm{K}^{-1}$).
At $10\;\mathrm{GHz}$ modulation rate, the thermal contribution amounts
to ${\sim}0.005$--$0.02\;\mathrm{pJ}$ per ring per operation.
For the $N = 30$ cascade, this sums to $0.15$--$0.60\;\mathrm{pJ}$,
which is comparable to $E_{\mathrm{EO}}$ and must be included in the
total: $E_{\mathrm{total}} \approx 0.94$--$1.39\;\mathrm{pJ}$.
The total energy comparison should therefore be treated as an
order-of-magnitude estimate.

\subsection*{S7.4~~Comparison with electronic implementations}

Here we provide an order-of-magnitude energy comparison between electrical
analog exponential circuits and our photonic MRR cascade, grounding
the analysis in published device data and first-principles estimates.
We assume a shared CW laser with total optical output
$P_\mathrm{in,tot} = 1\;\mathrm{mW}$, split across $M = 10$ parallel
softmax channels via a $1 \times M$ power splitter, yielding
per-channel input $P_\mathrm{in,ch} = 100\;\mu\mathrm{W}$.
The output power at the cascade drop port is
$P_\mathrm{out} = P_\mathrm{in,ch} \times D_\mathrm{max}^N$, which
ranges from $0.61\;\mu\mathrm{W}$ (FDTD regime, $N = 5$) to
$21.5\;\mu\mathrm{W}$ (fabricated regime, $N = 30$)
(Table~\ref{tab:power_budget}).

\paragraph{Electrical analog exponential circuits.}
Three main families of electrical analog exponential circuits are compared:
BJT translinear/Gilbert cell ($\sim 3$~pJ at 100~MHz~\cite{Gilbert1975,Mead1989,GilbertCell2009}),
CMOS subthreshold ($\sim 0.43$~pJ at 1~MHz~\cite{Mead1989,Softmax180nm2021}),
and digital FP32 Taylor series ($\sim 46$~pJ at 1~GHz~\cite{Horowitz2014}).

\paragraph{Photonic MRR cascade: single-channel energy.}
For $N = 30$ X-cut TFLN micro-ring resonators in the
self-consistent high-$Q$ regime ($Q_L = 25{,}200$),
the three energy components are EO tuning
($E_\mathrm{EO} = 0.22$~pJ), amortized laser
($E_\mathrm{laser} = 0.07$~pJ, shared across $M = 10$
parallel channels), and photodetector
($E_\mathrm{PD} = 0.50$~pJ), yielding
$E_\mathrm{photonic} = 0.79$~pJ.
Including thermal stabilization for $N = 30$ rings
($0.15$--$0.60\;\mathrm{pJ}$),
the total rises to $0.94$--$1.39\;\mathrm{pJ}$.
Notably, $E_\mathrm{EO} \propto 1/N$ since $b \propto 1/N$ from minimax optimization.

\paragraph{Single-channel comparison.}
Table~\ref{tab:exp_comparison} presents the comparison.
The photonic cascade at $N = 30$ achieves $0.79\;\mathrm{pJ}$
baseline---$3.8\times$ lower than the BJT Gilbert cell
($3\;\mathrm{pJ}$) and $58\times$ lower than digital FP32
($46\;\mathrm{pJ}$).
Including thermal stabilization ($0.94$--$1.39\;\mathrm{pJ}$),
the advantage over INT8 ($2.3\;\mathrm{pJ}$) is
$1.7$--$2.4\times$, while operating at $10\;\mathrm{GHz}$ bandwidth.
At fabricated $Q \ge 30{,}000$, $E_\mathrm{EO}$ drops to $0.16\;\mathrm{pJ}$ and $E_\mathrm{total} \approx 0.73\;\mathrm{pJ}$ (excluding thermal; Table~\ref{tab:energy_vs_Q}), recovering a $3.2\times$ advantage over INT8.
Subthreshold CMOS achieves the lowest energy ($0.43\;\mathrm{pJ}$)
but at $10{,}000\times$ lower bandwidth.

\begin{table}[h!]
\centering
\caption{Energy per exponential operation: single-channel comparison.}
\label{tab:exp_comparison}
\begin{ruledtabular}
\begin{tabular}{lccc}
Implementation & $E/\mathrm{op}$ (pJ) & Bandwidth & Notes \\
\hline
Digital FP32 (Taylor)                & ${\sim}46$  & 1\,GHz   & 10 FP MACs \\
BJT Gilbert cell                     & ${\sim}3$   & 100\,MHz & Analog \\
Digital INT8 (Taylor)                & ${\sim}2.3$ & 1\,GHz   & 10 INT MACs \\
\textbf{Photonic MRR ($N\!=\!30$)}   & $\mathbf{0.94}$--$\mathbf{1.39}$ & \textbf{10\,GHz} & \textbf{Analog$^\dagger$} \\
Subthreshold CMOS                    & ${\sim}0.43$ & 1\,MHz  & Analog \\
\end{tabular}
\end{ruledtabular}
{\footnotesize $^\dagger$0.79~pJ excluding thermal; 0.94--1.39~pJ including thermal. Self-consistent with fabricated high-$Q$ regime ($Q_L = 25{,}200$); see Supplementary Sec.~S7.}
\end{table}

\paragraph{Caveats.}
These values are order-of-magnitude estimates, not device-accurate predictions.
The photonic estimate excludes DAC energy for voltage generation
(typically $0.1$--$1\;\mathrm{pJ}$ per conversion at
$10\;\mathrm{GHz}$~bandwidth, shared with any analog approach) and
thermal tuning power for maintaining resonance alignment
(${\sim}50$--$200\;\mu\mathrm{W}$ per ring for TFLN, lower than
silicon due to the small thermo-optic coefficient of \ce{LiNbO3},
$dn/dT \approx 3.9 \times 10^{-6}\;\mathrm{K}^{-1}$).
Effective precision at the photodetector is limited to ${\sim}$6--8~bits
by shot noise and receiver electronics.
The energy advantage over electrical implementations is strongest in
the fabricated high-$Q$ regime ($D_\mathrm{max} \geq 0.95$), where
$N = 30$ is practical and $V_\mathrm{ctrl}$ remains CMOS-compatible.

\clearpage
\section*{S8. Monte Carlo robustness under device non-idealities}
This section describes the robustness model summarized in the main text.
For the fitted $L=8$, $N=10$ design ($a=-1.4588$, $b=0.10202$), each Monte Carlo chip sample includes:
(i) per-ring static detuning spread,
(ii) per-ring sensitivity spread,
(iii) global thermal drift and crosstalk-like slope drift,
(iv) stage insertion-loss variation,
(v) control-channel noise,
and (vi) detector noise with one-point calibration at $I=L$.

For ring $k$, we use
\begin{equation}
T_k(I)=\frac{1}{1+\left(a_k+b_k I + d_{\mathrm{th}} + d_{\mathrm{xt}}\,I/L\right)^2},
\end{equation}
with
\begin{equation}
y(I)=\prod_{k=1}^{N} T_k(I)\times 10^{-IL_{\mathrm{tot}}/10},
\end{equation}
and one-point calibration $\tilde y(I)=C_{\mathrm{cal}}y(I)$ such that $\tilde y(L)=1$ for the same chip instance.

\begin{table}[h!]
\centering
\caption{Non-ideality distributions used in the Monte Carlo sweeps.}
\label{tab:s1_params}
\begin{tabular}{lcc}
\toprule
Parameter & Nominal & Stress \\
\midrule
$\sigma_a$ & 0.020 & 0.032 \\
$\sigma_{b,\mathrm{rel}}$ & 0.020 & 0.032 \\
$\sigma_{\mathrm{th}}$ & 0.015 & 0.025 \\
$\sigma_{\mathrm{xt}}$ & 0.012 & 0.020 \\
$\sigma_I$ & 0.004 & 0.007 \\
$IL_{\mathrm{stage}}$ (dB, $\mu\pm\sigma$) & $0.12\pm0.03$ & $0.18\pm0.05$ \\
$\sigma_{\mathrm{det}}$ & $3.0\times 10^{-6}$ & $6.0\times 10^{-6}$ \\
\bottomrule
\end{tabular}
\end{table}

\begin{table}[h!]
\centering
\caption{Monte Carlo summary (same run reported in main text).}
\label{tab:s1_summary}
\begin{tabular}{lcc}
\toprule
Metric & Nominal & Stress \\
\midrule
Median KL$(p_{\mathrm{ref}}\|p_{\mathrm{approx}})$ & $2.17\times10^{-4}$ & $7.39\times10^{-4}$ \\
p95 KL$(p_{\mathrm{ref}}\|p_{\mathrm{approx}})$ & $5.92\times10^{-4}$ & $2.21\times10^{-3}$ \\
Median max $|\Delta p|$ & $0.170\%$ & $0.193\%$ \\
p95 max $|\Delta p|$ & $0.319\%$ & $0.419\%$ \\
\bottomrule
\end{tabular}
\end{table}

\begin{figure}[h!]
\centering
\includegraphics[width=0.82\linewidth]{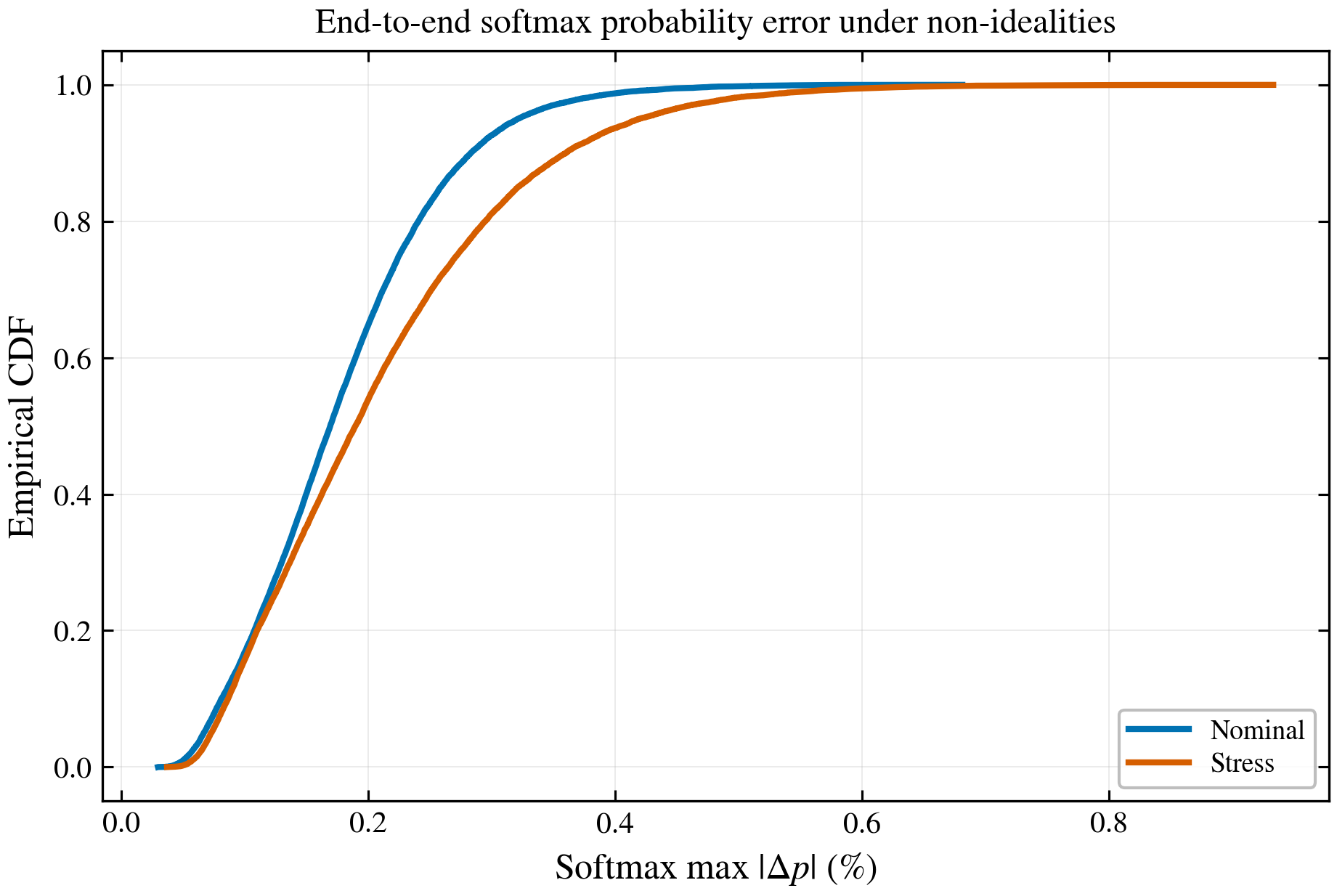}
\caption{CDF of end-to-end softmax probability error under the same non-ideality samples.}
\label{fig:s1_softmax_cdf}
\end{figure}

\noindent\textbf{Conservative-bound sketch used for the main-text screening equation.}
For the identical-detuning family with fixed $b$, define
\begin{equation}
\ln \tilde y(I)=N\,\phi(a+bI)-N\,\phi(a+bL),\qquad \phi(u)=-\ln(1+u^2),
\end{equation}
so that $\tilde y(L)=1$ by construction. Around a constructive choice with $a+bI<0$ on $[0,L]$, a second-order remainder argument for the mismatch between the target slope and the fitted slope yields a term scaling as $L^2/(4N)$, while the flank-curvature penalty contributes a term scaling as $1/(2b^2N)$. Combining the two contributions gives the screening inequality
\begin{equation}
E_\infty \lesssim \frac{L^2}{4N}+\frac{1}{2b^2N},
\end{equation}
which leads to the conservative screening equation reported in the main manuscript. We emphasize that this is a conservative heuristic design rule (not a formal minimax theorem), used only for preliminary depth screening.

\clearpage
\section*{S9. Delay-aware feedback normalization validation}
We model global normalization as a delayed PI-controlled loop:
\begin{align}
S(t) &= G(t)P(t)+n(t),\\
\tau\frac{dP}{dt} &= -P(t)+u(t-T_d),\\
u(t) &= K_p e(t)+K_i\int e(t)\,dt,\qquad e(t)=S_{\mathrm{ref}}-S(t),
\end{align}
with actuator saturation $0\le u\le P_{\max}$.
A piecewise $G(t)$ profile is used to emulate workload changes.
For physical intuition, Table~\ref{tab:s2_tau_examples} converts normalized delay/settling metrics into absolute-time examples.

\begin{table}[h!]
\centering
\caption{Example absolute-time interpretation of normalized PI-loop metrics using one representative stable case ($(K_p,K_i,T_d/\tau)=(0.55,0.8,0.2)$) and a $\pm2\%$ settling-time definition ($T_{\mathrm{settle}}\sim12.4\tau$).}
\label{tab:s2_tau_examples}
\begin{tabular}{cccc}
\toprule
Assumed $\tau$ & Delay $T_d=0.2\tau$ & Settling $\sim12.4\tau$ & Interpretation \\
\midrule
$100\,\mathrm{ns}$ & $20\,\mathrm{ns}$ & $1.24\,\mu\mathrm{s}$ & fast loop \\
$1\,\mu\mathrm{s}$ & $200\,\mathrm{ns}$ & $12.4\,\mu\mathrm{s}$ & moderate loop \\
$5\,\mu\mathrm{s}$ & $1\,\mu\mathrm{s}$ & $62\,\mu\mathrm{s}$ & slower loop \\
\bottomrule
\end{tabular}
\end{table}

\noindent\textbf{Reference-backed latency context for bottleneck screening.}
To place the delayed-loop times against mixed-signal system latencies, Table~\ref{tab:s2_tsys} summarizes representative time scales with explicit path classes (on-chip vs off-chip) for memory and interconnect paths, alongside conversion latency ranges. These are intentionally order-of-magnitude ranges (not fixed constants), and can shift with architecture, clocking, and protocol stack choices.

\begin{table}[h!]
\centering
\small
\caption{Representative subsystem latency ranges used for conservative bottleneck screening in Sec.~S9.}
\label{tab:s2_tsys}
\begin{tabular}{lll}
\toprule
Subsystem path & $T_{\mathrm{sys}}$ & Sources \\
\midrule
On-chip memory (L1/L2) & $20$--$200\,\mathrm{ns}$ & \cite{jia2018volta} \\
Off-chip memory (DRAM) & $200$--$700\,\mathrm{ns}$ & \cite{jia2018volta,kim2012salp} \\
ADC conversion & $10$--$710\,\mathrm{ns}$ & \cite{ti_adc12dj3200,ti_ads8881} \\
DAC + driver/settling & $1$--$200\,\mathrm{ns}$ & \cite{ti_dac38rf82} \\
On-chip interconnect (NoC) & $5$--$100\,\mathrm{ns}$ & \cite{dally2001route} \\
Off-chip I/O (PCIe/CXL) & $1$--$10\,\mu\mathrm{s}$ & \cite{jia2018volta,sano2023cxl} \\
\bottomrule
\end{tabular}
\end{table}

\noindent\textbf{Conservative risk-screening heuristic for loop latency.}
As a screening heuristic, we use the settling time from one representative stable case ($(K_p,K_i,T_d/\tau)=(0.55,0.8,0.2)$; Table~\ref{tab:s2_step}), with settling defined as the first time entering and remaining within a $\pm2\%$ band around $S_{\mathrm{ref}}$, as a normalization-loop latency proxy:
\begin{equation}
T_{\mathrm{norm}} \approx 12.4\,\tau.
\label{eq:s2_tnorm}
\end{equation}
This value is \emph{not} a universal bound; different gain settings, delay ratios, or loop architectures will yield different settling times.
It is used only as a reference point for order-of-magnitude risk screening.
Define the conservative screening metric
\begin{equation}
T_{\mathrm{norm}} \ge \beta\,T_{\mathrm{sys}},
\label{eq:s2_bottleneck_cond}
\end{equation}
with $\beta=1$ (high-risk screening line) and $\beta=0.5$ (early-warning line); this is a heuristic risk indicator, not a formal dominance proof. The corresponding threshold is
\begin{equation}
\tau_{\mathrm{crit}}(\beta)=\frac{\beta\,T_{\mathrm{sys}}}{12.4}.
\label{eq:s2_taucrit}
\end{equation}
Table~\ref{tab:s2_taucrit} gives the resulting numeric ranges.

\begin{table}[h!]
\centering
\caption{Computed $\tau_{\mathrm{crit}}$ ranges from Eq.~(\ref{eq:s2_taucrit}) using Table~\ref{tab:s2_tsys}.}
\label{tab:s2_taucrit}
\begin{tabular}{lccc}
\toprule
Subsystem & $T_{\mathrm{sys}}$ range & $\tau_{\mathrm{crit}}$ ($\beta=0.5$) & $\tau_{\mathrm{crit}}$ ($\beta=1$) \\
\midrule
On-chip memory path & $20$--$200$ ns & $0.81$--$8.06$ ns & $1.61$--$16.13$ ns \\
Off-chip memory path & $200$--$700$ ns & $8.06$--$28.23$ ns & $16.13$--$56.45$ ns \\
ADC conversion & $10$--$710$ ns & $0.40$--$28.63$ ns & $0.81$--$57.26$ ns \\
DAC+driver/settling & $1$--$200$ ns & $0.04$--$8.06$ ns & $0.08$--$16.13$ ns \\
On-chip interconnect (NoC) & $5$--$100$ ns & $0.20$--$4.03$ ns & $0.40$--$8.06$ ns \\
Off-chip I/O fabric & $1$--$10\,\mu$s & $0.04$--$0.40\,\mu$s & $0.08$--$0.81\,\mu$s \\
\bottomrule
\end{tabular}
\end{table}

\noindent For the explicit examples in Table~\ref{tab:s2_tau_examples},
$\tau=0.1\,\mu\mathrm{s}$ gives $T_{\mathrm{norm}}\approx1.24\,\mu\mathrm{s}$,
$\tau=1\,\mu\mathrm{s}$ gives $T_{\mathrm{norm}}\approx12.4\,\mu\mathrm{s}$,
and $\tau=5\,\mu\mathrm{s}$ gives $T_{\mathrm{norm}}\approx62\,\mu\mathrm{s}$.
These numbers indicate a risk trend (not a hard boundary): for this representative case, the normalization loop is typically non-dominant when $\tau$ is well below the relevant $\tau_{\mathrm{crit}}$ band, and it may become dominant as $\tau$ approaches or exceeds that band.
The transition depends on path class (on-chip vs off-chip) and on architecture-specific timing closure, including whether the normalization path lies on the end-to-end critical path (Table~\ref{tab:s2_tsys}).
\textbf{Accordingly, this analysis is intended for preliminary risk screening only; concrete implementations require full timing validation.}

\begin{table}[h!]
\centering
\caption{Representative step-response cases for the delayed PI loop (settling defined by a $\pm2\%$ band around $S_{\mathrm{ref}}$).}
\label{tab:s2_step}
\begin{tabular}{lcccc}
\toprule
Case & $(K_p,K_i,T_d/\tau)$ & Overshoot & Settling & Stable \\
\midrule
Stable & $(0.55,0.8,0.2)$ & $25.6\%$ & $\sim12.4\tau$ & Yes \\
Marginal & $(0.95,1.6,0.45)$ & $25.6\%$ & $\sim12.8\tau$ & Yes \\
Unstable & $(1.2,2.2,0.75)$ & $45.1\%$ & not settled & No \\
\bottomrule
\end{tabular}
\end{table}

\begin{table}[h!]
\centering
\caption{Stable-region fraction from gain-map scans at each delay ratio.}
\label{tab:s2_map}
\begin{tabular}{cc}
\toprule
$T_d/\tau$ & Stable fraction \\
\midrule
0.0 & 88.1\% \\
0.2 & 88.0\% \\
0.5 & 72.4\% \\
0.8 & 47.5\% \\
\bottomrule
\end{tabular}
\end{table}

\begin{figure}[h!]
\centering
\includegraphics[width=0.9\linewidth]{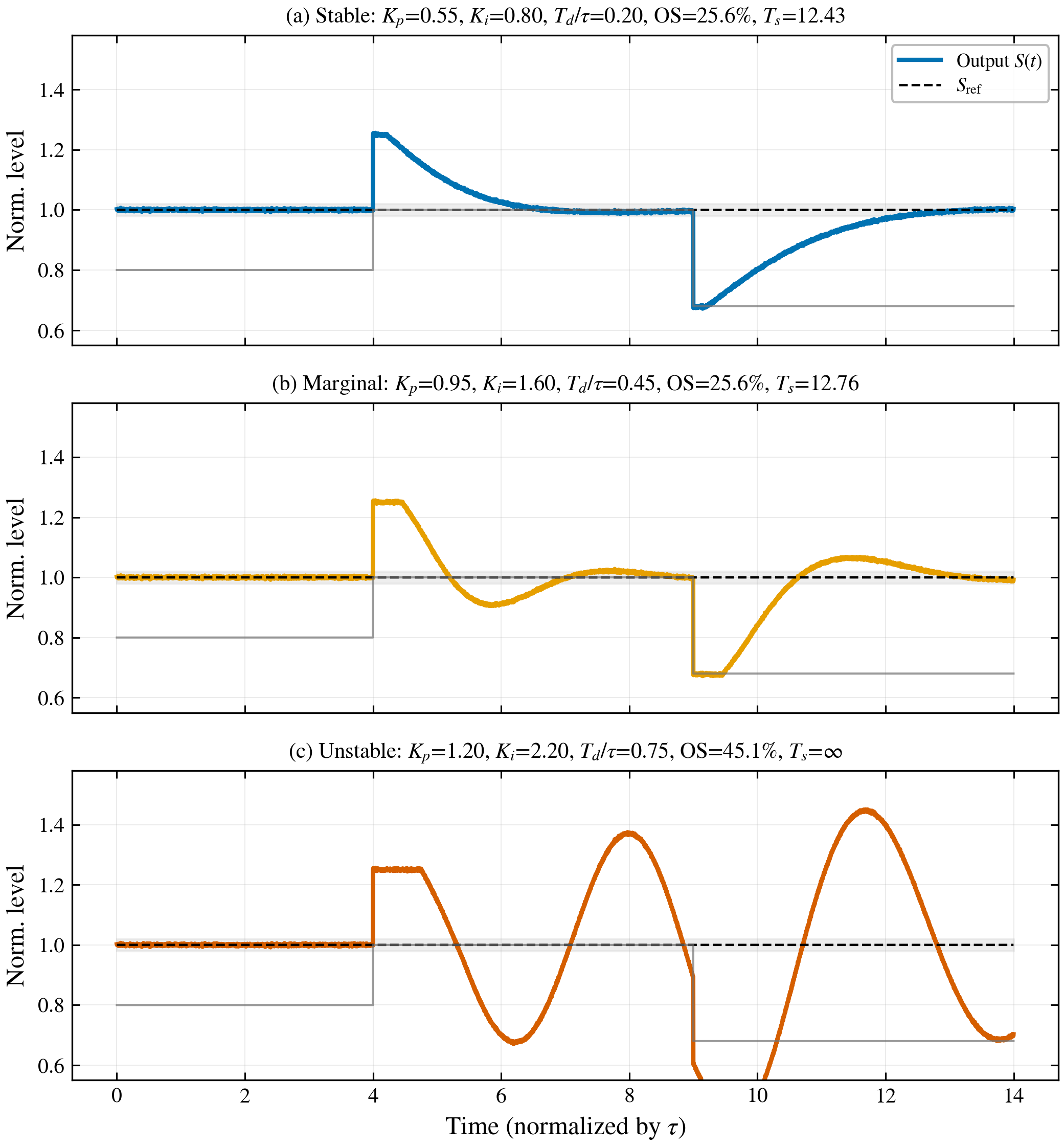}
\caption{Step-response examples of the delayed PI normalization loop.}
\label{fig:s2_step}
\end{figure}

\begin{figure}[h!]
\centering
\includegraphics[width=0.9\linewidth]{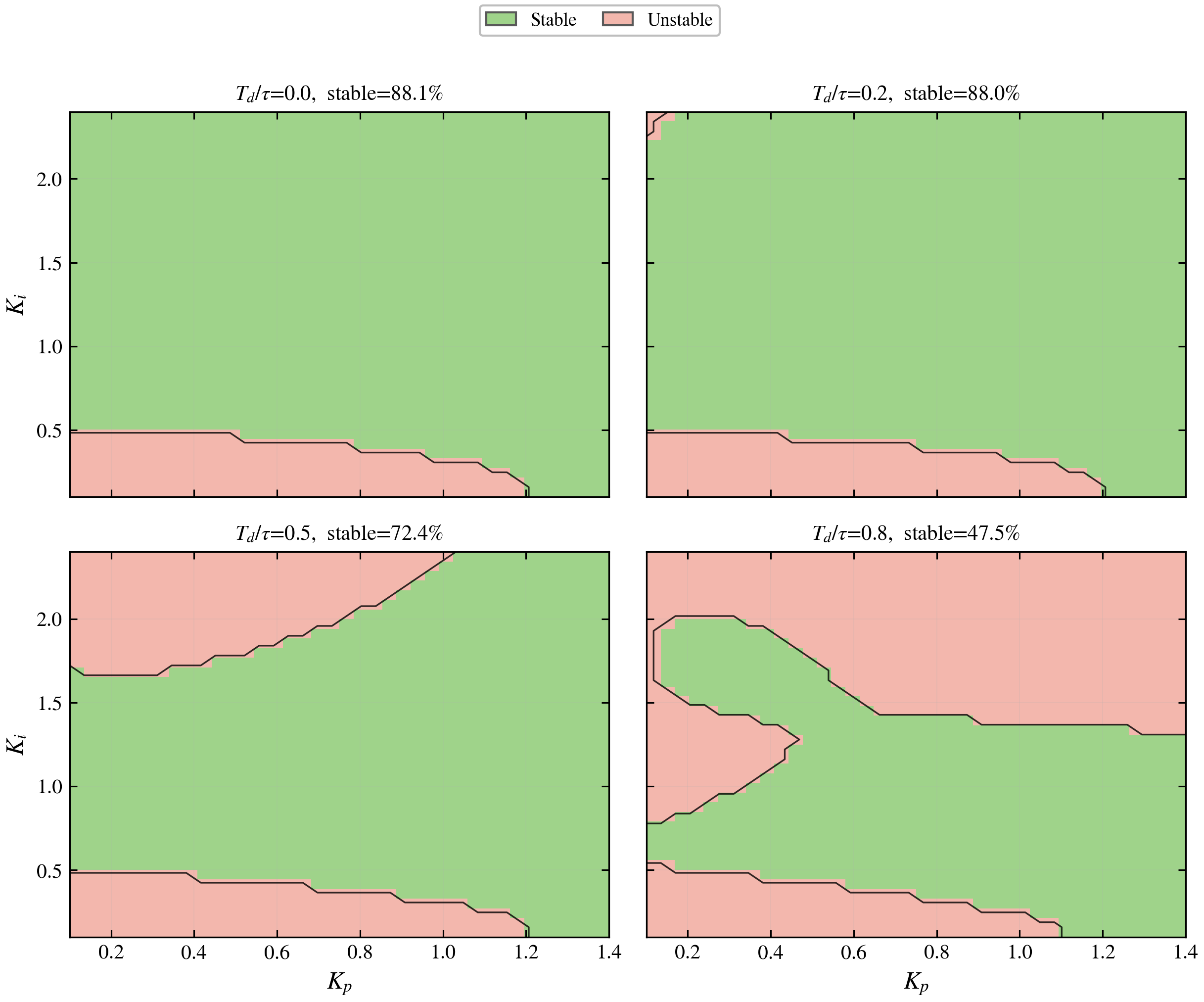}
\caption{Delay-dependent stability maps over scanned $(K_p,K_i)$ ranges.}
\label{fig:s2_map}
\end{figure}

\clearpage
\section*{S10. Reproducibility}
Scripts used for this Supplementary validation:
\begin{itemize}
\item \texttt{scripts/nonideality\_montecarlo.py}
\item \texttt{scripts/feedback\_loop\_validation.py}
\item \texttt{scripts/extract\_logit\_range\_effective.py}
\item \texttt{scripts/analyze\_softmax\_clipping\_validity.py}
\end{itemize}
Public code repository: \url{https://github.com/hyoseokp/MRR-AEF} (commit \texttt{585e695}).
Empirical extraction outputs are stored under:
\begin{itemize}
\item \texttt{paper/empirical\_L\_v3/}
\end{itemize}



\putbib[refs_new]

\end{bibunit}
\end{document}